\documentclass[english]{article}
\usepackage[T1]{fontenc}
\usepackage[latin9]{inputenc}
\usepackage{geometry}
\usepackage{babel}
\usepackage{amsthm}
\usepackage{amsmath}
\usepackage{amssymb}
\usepackage{esint}
\usepackage[unicode=true,pdfusetitle,
 bookmarks=true,bookmarksnumbered=false,bookmarksopen=false,
 breaklinks=false,pdfborder={0 0 1},backref=false,colorlinks=false]
 {hyperref}
\usepackage{breakurl}

\makeatletter

\providecommand{\tabularnewline}{\\}

\theoremstyle{plain}
\newtheorem{thm}{\protect\theoremname}
  \theoremstyle{definition}
  \newtheorem{defn}[thm]{\protect\definitionname}
  \theoremstyle{plain}
  \newtheorem{prop}[thm]{\protect\propositionname}
  \theoremstyle{plain}
  \newtheorem{cor}[thm]{\protect\corollaryname}
  \theoremstyle{plain}
  \newtheorem{lem}[thm]{\protect\lemmaname}

\usepackage{bbm}
\usepackage{graphicx}

\makeatother

  \providecommand{\corollaryname}{Corollary}
  \providecommand{\definitionname}{Definition}
  \providecommand{\lemmaname}{Lemma}
  \providecommand{\propositionname}{Proposition}
\providecommand{\theoremname}{Theorem}

\begin{document}
\global\long\def\one{\mathbbm{1}}
\global\long\def\id{\mathrm{id}}
\global\long\def\diag{\mathrm{diag}}
\global\long\def\C{\mathbb{C}}
\global\long\def\R{\mathbb{R}}
\global\long\def\Z{\mathbb{Z}}
\global\long\def\N{\mathbb{N}}
\global\long\def\ot{\otimes}
\global\long\def\L{\mathcal{L}}

\global\long\def\Tr#1{\mathrm{Tr}\left(#1\right)}
\global\long\def\bra#1{\langle#1|}
\global\long\def\ket#1{|#1\rangle}
\global\long\def\braket#1#2{\left\langle #1\middle|#2\right\rangle }
\global\long\def\Mtwo#1#2#3#4{\begin{pmatrix}#1  &  #2\\
#3  &  #4 
\end{pmatrix}}

\global\long\def\spec{\mathrm{Spec}\,}
\global\long\def\re{\mathrm{\, Re}\,}
\global\long\def\im{\,\mathrm{Im}\,}
\global\long\def\dt{\, dt}
\global\long\def\dy{\, dy}
\global\long\def\dz{\, dz}
\global\long\def\no#1{\left\Vert #1\right\Vert }
\global\long\def\sV{\Theta}
\global\long\def\sG{\Gamma}

\title{Analysis of energy transfer in quantum networks using kinetic network
approximations}

\author{David K. Moser\textsuperscript{1,2}{\small }\\
{\small }\\
{\small 1: Department of Mathematics}\\
{\small 2: Department of Physics}\\
{\small Northeastern University}\\
{\small Boston MA 02115}\\
{\small david.moser@gmx.net}}
\maketitle
\begin{abstract}
Coherent energy transfer in pigment-protein complexes has been studied
by mapping the quantum network to a kinetic network. This gives an
analytic way to find parameter values for optimal transfer efficiency.
In the case of the Fenna-Matthews-Olson (FMO) complex, the comparison
of quantum and kinetic network evolution shows that dephasing-assisted
energy transfer is driven by the two-site coherent interaction, and
not system-wide coherence. Using the Schur complement, we find a new
kinetic network that gives a closer approximation to the quantum network
by including all multi-site coherence contributions. Our new network
approximation can be expanded as a series with contributions representing
different numbers of coherently interacting sites.

For both kinetic networks we study the system relaxation time, the
time it takes for the excitation to spread throughout the complex.
We make mathematically rigorous estimates of the relaxation time when
comparing kinetic and quantum network. Numerical simulations comparing
the coherent model and the two kinetic network models, confirm our
bounds, and show that the relative error of the new kinetic network
approximation is several orders of magnitude smaller.

\end{abstract}
\tableofcontents{}

\section{Introduction}

Since coherent energy transfer in the Fenna-Matthews-Olson complex
(FMO) has been observed \cite{engel_evidence_2007,panitchayangkoon_long-lived_2010,wong_quantum-coherent_2010},
extensive experimental and theoretical research has been dedicated
to studying coherent resonant transfer \cite{clegg_foerster_2010}
and the coherent pigment-protein interaction\cite{scholes_quantum-coherent_2010,ishizaki_quantum_2010}.
In particular, numerical solutions of simple models have shown that
dephasing -- the destruction of the coherences -- at an intermediate
rate helps to increase the energy transfer efficiency \cite{plenio_dephasing-assisted_2008,rebentrost_environment-assisted_2009}.
This has been called dephasing- or environment-assisted energy transfer,
and is analogous to a critically damped oscillator. The dephasing
corresponds to damping and causes the exciton to relax to an equal
distribution for every pigment site instead of staying localized due
to the energy mismatch between the sites.

The models are based on two assumptions. First, only a single exciton
is present, it is located at any of the seven pigments. The pigment
exciton energy, and the pigment dipole-dipole interaction \cite{cho_exciton_2005,adolphs_how_2006}
then lead to an oscillatory evolution of the system. And second, the
site-environment interactions are assumed to be purely Markovian without
any temporal or spatial correlations. The environment interactions
are dephasing, recombination and trapping. Dephasing destroys the
site coherences without destroying the exciton itself, and phonon
recombination or photon re-emission lead to loss of the exciton to
the environment. Trapping is the transfer of the exciton to the reaction
center, where the electronic energy is converted to chemical energy,
in FMO it occurs at pigment 3. The \emph{transfer efficiency} is the
probability that an exciton starting at site 1 or site 6 reaches the
reaction center. For a general system with $n$ pigments, we convert
the master equation of the coherent model into vector form
\[
\dot{\vec{\rho}}=M\vec{\rho}
\]
where $\vec{\rho}\in\R^{n^{2}}$ is the density matrix in vector form
and $M$ is a real $n^{2}\times n^{2}$-matrix. Two procedures to
find $M$ are presented in \ref{sub:Converting-the-master} and \ref{sub:Numerical-simulations}.

To study population transfer channels and conditions for optimal transfer,
a mapping to kinetic networks has been proposed \cite{cao_optimization_2009,hoyer_limits_2010}.
A kinetic network is a system where the exciton jumps incoherently
between sites according to some fixed rates, i.e. a continuous-time
Markov process. In its simplest version this approximation only takes
into account the coherent interaction between pairs of sites to derive
the transfer rate between them. If the two sites interact with strength
$V$, have an energy separation $E$, and both sites experience dephasing
at rate $\gamma$ and population loss at rate $\kappa$ then the rate
is
\begin{equation}
\mu=\frac{2\left|V\right|^{2}(\gamma+\kappa)}{(\gamma+\kappa)^{2}+E^{2}}\,.\label{eq:intro-rate}
\end{equation}
This rate is maximized for the intermediate dephasing rate $\gamma=E-\kappa$
so the phenomenon of dephasing-assisted transfer is maintained in
this approximation. For a system with $n$ sites, these rates constitute
the off-diagonals of a $n\times n$ rate matrix $N_{0}$, and the
system populations evolve according to
\[
\dot{\vec{p}}=N_{0}\vec{p}
\]
where $\vec{p}\in\R^{n}$ is the time-dependent population vector.
Figure~\ref{fig:FMO-efficiency} displays the transfer efficiency
with models $M$ and $N_{0}$ for different $\gamma$, the dephasing-assisted
regime clearly shows as a peak around $\gamma\approx170cm^{-1}$.
At the peak the population evolution of $M$ is well approximated
by that of $N_{0}$, therefore dephasing-assisted energy transfer
can be explained by the relatively simple coherent dynamic between
pairs of sites that enters the rate $\mu$ and the influence of system-wide
coherence is small.

\begin{figure}
\begin{centering}
\begingroup
  \makeatletter
  \providecommand\color[2][]{    \GenericError{(gnuplot) \space\space\space\@spaces}{      Package color not loaded in conjunction with
      terminal option `colourtext'    }{See the gnuplot documentation for explanation.    }{Either use 'blacktext' in gnuplot or load the package
      color.sty in LaTeX.}    \renewcommand\color[2][]{}  }  \providecommand\includegraphics[2][]{    \GenericError{(gnuplot) \space\space\space\@spaces}{      Package graphicx or graphics not loaded    }{See the gnuplot documentation for explanation.    }{The gnuplot epslatex terminal needs graphicx.sty or graphics.sty.}    \renewcommand\includegraphics[2][]{}  }  \providecommand\rotatebox[2]{#2}  \@ifundefined{ifGPcolor}{    \newif\ifGPcolor
    \GPcolortrue
  }{}  \@ifundefined{ifGPblacktext}{    \newif\ifGPblacktext
    \GPblacktexttrue
  }{}    \let\gplgaddtomacro\g@addto@macro
    \gdef\gplbacktext{}  \gdef\gplfronttext{}  \makeatother
  \ifGPblacktext
        \def\colorrgb#1{}    \def\colorgray#1{}  \else
        \ifGPcolor
      \def\colorrgb#1{\color[rgb]{#1}}      \def\colorgray#1{\color[gray]{#1}}      \expandafter\def\csname LTw\endcsname{\color{white}}      \expandafter\def\csname LTb\endcsname{\color{black}}      \expandafter\def\csname LTa\endcsname{\color{black}}      \expandafter\def\csname LT0\endcsname{\color[rgb]{1,0,0}}      \expandafter\def\csname LT1\endcsname{\color[rgb]{0,1,0}}      \expandafter\def\csname LT2\endcsname{\color[rgb]{0,0,1}}      \expandafter\def\csname LT3\endcsname{\color[rgb]{1,0,1}}      \expandafter\def\csname LT4\endcsname{\color[rgb]{0,1,1}}      \expandafter\def\csname LT5\endcsname{\color[rgb]{1,1,0}}      \expandafter\def\csname LT6\endcsname{\color[rgb]{0,0,0}}      \expandafter\def\csname LT7\endcsname{\color[rgb]{1,0.3,0}}      \expandafter\def\csname LT8\endcsname{\color[rgb]{0.5,0.5,0.5}}    \else
            \def\colorrgb#1{\color{black}}      \def\colorgray#1{\color[gray]{#1}}      \expandafter\def\csname LTw\endcsname{\color{white}}      \expandafter\def\csname LTb\endcsname{\color{black}}      \expandafter\def\csname LTa\endcsname{\color{black}}      \expandafter\def\csname LT0\endcsname{\color{black}}      \expandafter\def\csname LT1\endcsname{\color{black}}      \expandafter\def\csname LT2\endcsname{\color{black}}      \expandafter\def\csname LT3\endcsname{\color{black}}      \expandafter\def\csname LT4\endcsname{\color{black}}      \expandafter\def\csname LT5\endcsname{\color{black}}      \expandafter\def\csname LT6\endcsname{\color{black}}      \expandafter\def\csname LT7\endcsname{\color{black}}      \expandafter\def\csname LT8\endcsname{\color{black}}    \fi
  \fi
  \setlength{\unitlength}{0.0500bp}  \begin{picture}(7200.00,5760.00)    \gplgaddtomacro\gplbacktext{    }    \gplgaddtomacro\gplfronttext{    }    \gplgaddtomacro\gplbacktext{      \csname LTb\endcsname      \put(588,576){\makebox(0,0)[r]{\strut{}-8}}      \put(588,1152){\makebox(0,0)[r]{\strut{}-6}}      \put(588,1728){\makebox(0,0)[r]{\strut{}-4}}      \put(588,2303){\makebox(0,0)[r]{\strut{}-2}}      \csname LTb\endcsname      \put(720,356){\makebox(0,0){\strut{}-3}}      \csname LTb\endcsname      \put(1440,356){\makebox(0,0){\strut{}-2}}      \csname LTb\endcsname      \put(2160,356){\makebox(0,0){\strut{}-1}}      \csname LTb\endcsname      \put(2880,356){\makebox(0,0){\strut{}0}}      \csname LTb\endcsname      \put(3600,356){\makebox(0,0){\strut{}1}}      \csname LTb\endcsname      \put(4320,356){\makebox(0,0){\strut{}2}}      \csname LTb\endcsname      \put(5040,356){\makebox(0,0){\strut{}3}}      \csname LTb\endcsname      \put(5760,356){\makebox(0,0){\strut{}4}}      \csname LTb\endcsname      \put(6480,356){\makebox(0,0){\strut{}5}}      \colorrgb{0.00,0.00,0.00}      \put(82,1727){\rotatebox{90}{\makebox(0,0){\strut{}$\log_{10}\frac{\Delta f}{f}$}}}      \colorrgb{0.00,0.00,0.00}      \put(3600,26){\makebox(0,0){\strut{}$\log_{10}\gamma\,(cm^{-1})$}}    }    \gplgaddtomacro\gplfronttext{    }    \gplgaddtomacro\gplbacktext{      \csname LTb\endcsname      \put(588,2880){\makebox(0,0)[r]{\strut{}0.7}}      \put(588,3648){\makebox(0,0)[r]{\strut{}0.8}}      \put(588,4415){\makebox(0,0)[r]{\strut{}0.9}}      \put(588,5183){\makebox(0,0)[r]{\strut{}1}}      \colorrgb{0.00,0.00,0.00}      \put(720,2660){\makebox(0,0){\strut{}}}      \colorrgb{0.00,0.00,0.00}      \put(1440,2660){\makebox(0,0){\strut{}}}      \colorrgb{0.00,0.00,0.00}      \put(2160,2660){\makebox(0,0){\strut{}}}      \colorrgb{0.00,0.00,0.00}      \put(2880,2660){\makebox(0,0){\strut{}}}      \colorrgb{0.00,0.00,0.00}      \put(3600,2660){\makebox(0,0){\strut{}}}      \colorrgb{0.00,0.00,0.00}      \put(4320,2660){\makebox(0,0){\strut{}}}      \colorrgb{0.00,0.00,0.00}      \put(5040,2660){\makebox(0,0){\strut{}}}      \colorrgb{0.00,0.00,0.00}      \put(5760,2660){\makebox(0,0){\strut{}}}      \colorrgb{0.00,0.00,0.00}      \put(6480,2660){\makebox(0,0){\strut{}}}      \colorrgb{0.00,0.00,0.00}      \put(-50,4031){\rotatebox{90}{\makebox(0,0){\strut{}$f$}}}    }    \gplgaddtomacro\gplfronttext{      \csname LTb\endcsname      \put(2568,5010){\makebox(0,0)[r]{\strut{}quantum}}      \csname LTb\endcsname      \put(2568,4790){\makebox(0,0)[r]{\strut{}kinetic $N_0$}}      \csname LTb\endcsname      \put(2568,4570){\makebox(0,0)[r]{\strut{}kinetic $N$}}    }    \gplbacktext
    \put(0,0){\includegraphics{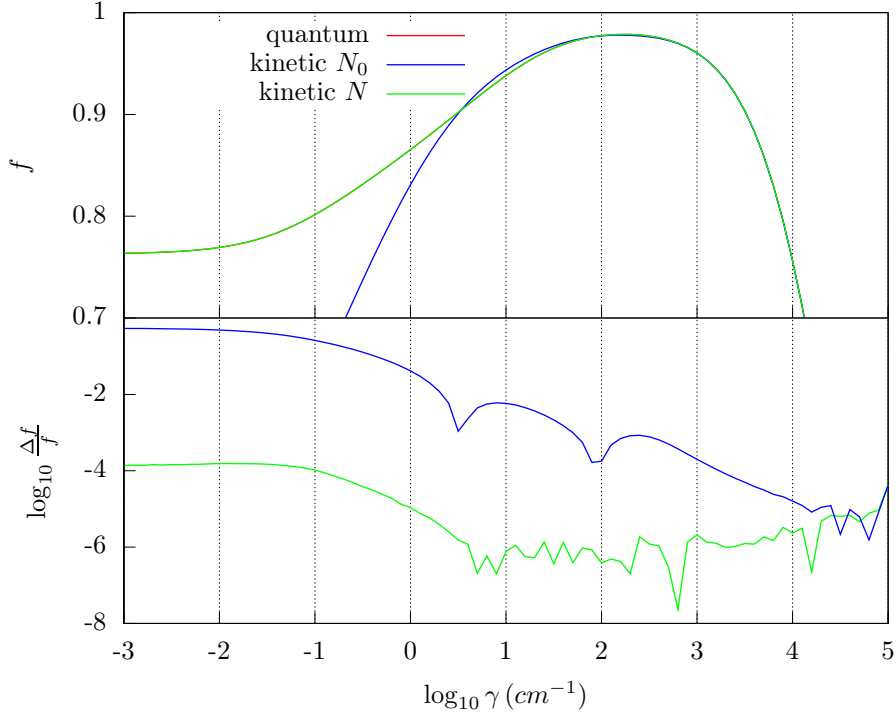}}    \gplfronttext
  \end{picture}\endgroup
\par\end{centering}

\caption{The efficiency of energy transfer in the FMO monomer. Model parameters
are described in \ref{sub:The-FMO-complex}\label{fig:FMO-efficiency}}

\end{figure}

To extract the limit of good approximation we introduce scaling variables,
$\sG$ which is proportional to the energy separations, dephasing
and population loss rates, and $\sV$ which is proportional to the
site interactions. We will show that the approximation of $N_{0}$
to $M$ becomes good as $\sV\sG^{-1}$ approaches 0. We generalize
the procedure of finding a kinetic network approximation in a mathematically
appealing way using block matrices. We find a kinetic network matrix
$N$ that follows the evolution of $M$ much closer -- it is over
three orders of magnitude more precise than the network $N_{0}$ as
shown in Figure~\ref{fig:FMO-efficiency}. Further, it can be expanded
in $\sV\sG^{-1}$ as
\[
N=\sum_{k=0}^{\infty}N_{k}
\]
where $N_{0}$ is the approximation described above, and the $N_{k}$
are rate corrections due to coherent interactions via $k$ intermediate
sites. The expansion terms become smaller for increasing $k$, $N_{k}\propto\sV\cdot(\sV\sG^{-1})^{k}$.
By stopping the expansion at a finite $k$ kinetic networks approximation
of varying accuracy can be formed allowing the study of coherent interaction
at different ``scales'' or number of involved sites. We restrict
our further investigation to the dominant contribution $N_{0}$ and
the entire sum $N$.

In our exact bounds we study the system with all population-loss mechanisms
removed. Due to dephasing the exciton spreads throughout the system
at the \emph{exciton relaxation time} $\tau$ and all populations
become equal. The difference $\Delta\tau$ between relaxation times
of $M$ and $N$ or $N_{0}$ gives a simple measure of how good the
kinetic networks approximate the quantum network. As $\sV\sG^{-1}$
becomes small the kinetic networks approach the quantum network and
$\Delta\tau$ becomes small as well.

We define $\tau$ and $\Delta\tau$ as follows, using the Euclidean
norm $\left\Vert \vec{p}\right\Vert _{2}=\sqrt{\sum_{i=1}^{n}p_{i}^{2}}$
to compare population vectors.
\begin{defn}
~
\begin{enumerate}
\item The map $T:\R^{n^{2}}\to\R^{n}$ is the restriction of density vectors
$\vec{\rho}$ to population vectors $\vec{p}$, and consequently $T^{\dagger}$
gives the embedding of population vector space in density vector space.
In particular, if the first $n$ components of $\vec{\rho}$ represent
the site populations, then $T=(\one_{n},0_{n\times(n^{2}-n)})$.
\item \label{enu:The-maximum-relaxation}The maximum relaxation time is
\[
\tau=\max_{\vec{p}_{0}}\left\Vert \int_{0}^{\infty}e^{Nt}\vec{p}_{0}-\frac{1}{n}\dt\right\Vert _{2}
\]
and the corresponding minimal relaxation rate is
\[
\mu=\frac{1}{\tau}
\]

\item \label{enu:The-maximum-deviation}The maximum deviation of relaxation
time between the quantum network $M$ and the kinetic network $N$
is
\[
\Delta\tau=\max_{\vec{p}_{0}}\left\Vert \int_{0}^{\infty}\left(Te^{Mt}T^{\dagger}-e^{Nt}\right)\vec{p}_{0}\dt\right\Vert _{2}\,.
\]

\item Define $\tau_{0}$, $\mu_{0}$ and $\Delta\tau_{0}$ in the same way,
replacing $N$ with $N_{0}$.
\end{enumerate}
\end{defn}
For our bounds we require that every site experiences dephasing. Further,
the network has to be \emph{connected}, meaning that any two sites
can exchange populations -directly or indirectly- such that the relaxed
state will have equal population everywhere. And finally we also require
our site interactions to be real -- but it is clear from our proofs
that the generalization to complex interactions could be treated in
a similar manner.

Our first results shows how fast the relaxation time of the two kinetic
networks $N_{0}$ and $N$ approximate that of the quantum network
$M$ as $\sV\sG^{-1}$ gets small.
\begin{thm}
\label{thm:intro-relax-bound}There are scaling invariant constants
$k_{1}$ and $k_{2}$, such that for $\sV\sG^{-1}$ small enough we
have the following bounds:
\begin{enumerate}
\item The relative difference of relaxation time between quantum evolution
$M$ and kinetic evolution $N_{0}$ is bounded by 
\[
\Delta\tau_{0,\,\mathrm{rel}}=\Delta\tau_{0}/\tau_{0}\leq k_{1}\sV\sG^{-1}\,.
\]

\item The relative difference of relaxation time between quantum evolution
$M$ and kinetic evolution $N$ is bounded by
\[
\Delta\tau_{\mathrm{rel}}=\Delta\tau/\tau\leq k_{2}\sV^{2}\sG^{-2}\,.
\]

\end{enumerate}
\end{thm}
This Theorem follows from Theorem~\ref{thm:relax-bound-1} and Corollary~\ref{cor:relax-bound}
in Section~\ref{sec:Bounding-relaxation-time}.

We also find the following exponential bounds on the time dependence.
\begin{thm}
There are scaling invariant constants $k_{3}$, $k_{4}$ and $k_{5}$,
such that for any initial population distribution $\vec{p}_{0}$ we
have the following bounds, as long as $\sV\sG^{-1}$ is small enough:
\begin{enumerate}
\item For all times $t\geq0$ 
\[
\left\Vert Te^{Mt}T^{\dagger}\vec{p}_{0}-e^{N_{0}t}\vec{p_{0}}\right\Vert _{2}\leq k_{3}e^{-\mu_{0}t/2}\cdot\sV\sG^{-1}\,.
\]

\item For all times $t\geq0$
\[
\left\Vert Te^{Mt}T^{\dagger}\vec{p}_{0}-e^{Nt}\vec{p_{0}}\right\Vert _{2}\leq k_{4}e^{-\mu t/2}\cdot\sV^{2}\sG^{-2}(1+k_{5}\log\sV\sG^{-1})\,.
\]

\end{enumerate}
\end{thm}
This Theorem follows from Theorem~\ref{thm:evolution-bound-1} and
Corollary~\ref{cor:evolution-bound} in Section~\ref{sec:Bounding-evolution-error}.
We expect that more sophisticated methods might yield the same bound
without the $\sV^{2}\sG^{-2}\log\sV\sG^{-1}$ term.

\section{The quantum network}

We first introduce the Master equation for the coherent model. Then
we reformulate the equation in vector form and combine the entire
dynamic in the real $n^{2}\times n^{2}$-matrix $M$. We describe
the general structure of $M$ as a preparation to the next section,
where we generate kinetic networks from parts of $M$.

\subsection{Master equation}

We consider the same quantum mechanical system studied in \cite{rebentrost_environment-assisted_2009}
with $n$ sites carrying a single excitation which is equivalent to
a system with $n$ states/levels. The site energies are $E_{k}\in\R$
so the energy operator is
\[
H=\sum_{k=1}^{n}E_{k}\ket k\bra k\,.
\]
The site $k$ couples to site $l$ with interaction strength $V_{kl}\in\C$
so the interaction operator is
\[
V=\sum_{k\neq l}V_{kl}\ket k\bra l
\]
 where $V_{kl}=\overline{V}_{lk}$. Site trapping, re-emission and
recombination can be incorporated by an anti-hermitian operator $A$.
Let $\kappa_{k}$ be the combined rate of exciton loss at site $k$
due to these effects, then $A$ is defined as
\[
A=\frac{-i}{2}\sum_{k=1}^{n}\kappa_{k}\ket k\bra k\,.
\]
Finally, every site is also under the influence of dephasing at rate
$\gamma_{k}\ge0$ incorporated in the Lindbladian superoperator
\[
\L(\rho)=\sum_{k=1}^{n}L_{k}\rho L_{k}^{\dagger}-\frac{1}{2}\{\rho,L_{k}^{\dagger}L_{k}\}
\]
with $L_{k}=\sqrt{\gamma_{k}}\ket i\bra i$. Setting $\hbar=1$, the
single exciton manifold of the quantum network is described by the
master equation 
\begin{equation}
\dot{\rho}=-i[H+V,\rho]-i\{A,\rho\}+\mathcal{L}(\rho)\label{eq:general-master-equation}
\end{equation}
where square and curly brackets represent commutator and anti-commutator
respectively.

For now we set $A=0$, ignoring exciton depleting processes as explained
above. We will mention how to include them in the kinetic network
approximations later on. Our approximation becomes exact in the limit
where the energy difference between sites is large, the dephasing
is large and the interactions are small. To be specific, we introduce
scaling parameters $\sG$ and $\sV$ and consider the limit $\sV\sG^{-1}\to0$.
Energies and dephasing scale like $\sG$ and interactions scale like
$\sV$
\begin{align*}
E_{k} & \propto\sG\,,\\
\gamma_{k} & \propto\sG\,,\\
V_{kl} & \propto\sV\,.
\end{align*}

With these assumptions the master equation turns into
\begin{equation}
\dot{\rho}=-i[\sG H+\sV V,\rho]+\sG\L(\rho)\,.\label{eq:master-equation-no-anti}
\end{equation}
Because this equation is linear in $\rho$ it can be converted into
vector form
\[
\dot{\vec{\rho}}=M\vec{\rho}
\]
where $\vec{\rho}\in\R^{n^{2}}$ is the density matrix in vector form
and $M$ is a real $n^{2}\times n^{2}$-matrix. Two procedures to
find $M$ are presented in \ref{sub:Converting-the-master} and \ref{sub:Numerical-simulations}.

\subsection{Converting to vector equation\label{sub:Converting-the-master}}

We rewrite the master equation (\ref{eq:master-equation-no-anti}),
skipping the scaling factors $\sV$ and $\sG$, it is easy to reintroduce
them at a later point
\begin{equation}
\dot{\rho}=-i[H+V,\rho]+\L(\rho)\,.\label{eq:master-equation-kinetic-networks}
\end{equation}
Our first goal is to convert this into the differential equation
\[
\dot{\vec{\rho}}=M\vec{\rho}
\]
for density ``vectors'' $\vec{\rho}\in\R^{n^{2}}$. Notice that
because $\rho=\rho^{\dagger}$ the space of density matrix has $n^{2}$
real dimensions, so we are not using any information when mapping
$\rho$ to $\vec{\rho}$.

We use the following conversion:
\begin{enumerate}
\item The first $n$ entries of the density vector are the populations --
the real diagonal entries of $\rho$.
\item For the entries $n+1$ to $n^{2}$ we alternate between real and imaginary
parts of the coherences -- the off-diagonal entries of $\rho$ --
starting with entry $\rho_{kl}$ where $k=1$ and $l=2$ continuing
by increasing $l$ until $l=n$, then moving to the entry $\rho_{23}$.
We multiply all these entries with $\sqrt{2}$, a normalization factor
useful to achieve simpler expressions later on.
\end{enumerate}
In terms of index equations this is:
\begin{enumerate}
\item For $k=1\dots n$
\[
\vec{\rho}_{k}=\rho_{kk}\,.
\]

\item For $k,l\in\{1,\dots,n\}$ with $k<l$
\begin{align*}
\vec{\rho}_{n+2n(k-1)-k(k+1)+2l-1} & =\sqrt{2}\re\rho_{kl}\\
\vec{\rho}_{n+2n(k-1)-k(k+1)+2l\phantom{-1}} & =\sqrt{2}\im\rho_{kl}\,.
\end{align*}

\end{enumerate}
Other mappings will yield the same kinetic networks, as long as they
allow for an easy separation of population and coherence space.

While somewhat tedious, it is now relatively straightforward to find
the matrix $M$ such that
\[
\dot{\vec{\rho}}=M\vec{\rho}\,.
\]
To find the rows $k=1\dots n$ we write out the diagonal components
of the RHS of (\ref{eq:master-equation-kinetic-networks}), and to
find rows $k=n+1,\dots,n^{2}$ we write out the off-diagonals of the
RHS of (\ref{eq:master-equation-kinetic-networks}). We follow this
procedure explicitly for the case $n=3$ in Appendix~\ref{app:Three-sites}.
From there it is obvious how the procedure generalizes to larger $n$.
Here we will only present the final form.

\subsection{The coherent evolution matrix $M$\label{sub:M}}

For simple notation and to simply extract the kinetic networks we
split up the density vector space $\R^{n^{2}}$. Let $P=\R^{n}$ be
the space of populations and let $C=\R^{n^{2}-n}$ be the space of
coherences. We can then write density vectors as $\vec{\rho}=\begin{pmatrix}\vec{p}\\
\vec{c}
\end{pmatrix}$ with $\vec{p}\in P$ and $\vec{c}\in C$. With this splitting the
matrix $M$ describing the quantum network looks like
\[
M=\begin{pmatrix}0 & -a^{\dagger}\\
a & b
\end{pmatrix}
\]
where $a:P\to C$ and $b:C\to C$ are real matrices (so $a^{\dagger}=a^{\top}$,
but we'll keep the more general notation for later). Notice that the
populations do not affect each other directly, but only via the coherences.

Matrix $a$ describes how populations couple to coherences, its entries
are real and imaginary parts of $V_{kl}$, naturally, site $k$ will
only couple to coherences $kl$ with $l\neq k$, thus of the $(n^{2}-n)$
entries in the $k$-th column of $a$ only $2(n-1)$ are nonzero.
Matrix $b$ describes how coherences couple to other coherences, if
considered as a block matrix with $2\times2$-blocks the diagonal
block for the coherence between site $k$ and site $l$ is of the
form
\begin{equation}
\begin{pmatrix}-\gamma_{kl} & -E_{kl}\\
E_{kl} & -\gamma_{kl}
\end{pmatrix}\label{eq:b-diagonal-blocks}
\end{equation}
where 
\begin{align*}
\gamma_{kl} & =\frac{1}{2}(\gamma_{k}+\gamma_{l})\\
E_{kl} & =E_{k}-E_{l}\,.
\end{align*}
 The off-diagonal blocks consist of real and imaginary parts of $V_{kl}$.

From the form of $M$, when ignoring the off-diagonal blocks of $b$,
we see that the site $k$ couples to the site $l$ via the coupling
strength $V_{kl}$, then some mixture of $\gamma_{kl}$ and $E_{kl}$
and then again via the coupling strength $V_{kl}$. This reminds us
of the rates of the form $\mu=\frac{2V^{2}\gamma}{\gamma^{2}+E^{2}}$
described in (\ref{eq:intro-rate}) that make up the matrix $N_{0}$.
We will make this intuition precise is the next subsections.

\section{Kinetic networks\label{sec:Kinetic-networks}}

In this section we show how the kinetic network $N$ emerges naturally
out of the study of the resolvent $(z-M)^{-1}$. We expand $N$ in
powers of $\sV\sG^{-1}$, giving the series
\[
N=\sum_{k=0}^{\infty}N_{k}
\]
with the leading order contribution being $N_{0}$. For some steps
involving matrix calculations we only give a simplified version. However,
in Appendix~(\ref{app:Three-sites}) we follow the procedure described
below, giving the full expressions in the case $n=3$.

\subsection{Extracting the kinetic network $N$\label{sub:kinetic-network-N}}

To extract kinetic networks from $M$ we consider its resolvent $(z-M)^{-1}$.
Remember that for any holomorphic function $f$ we have
\[
f(M)=\frac{1}{2\pi i}\oint f(z)(z-M)^{-1}\dz\,.
\]
Therefore, if one can bound the resolvent appropriately, one can also
bound the evolution operator $e^{Mt}$ and other related quantities.
Because we only care about approximating the population dynamics we
restrict our view to the population block of the resolvent of $M$.
The Banchiewicz formula \cite{banachiewicz_zur_1937} gives the inverse
of a $2\times2$-block matrix. The first block of the inverse -- in
our case the population block -- is called the Schur complement, and
due to its basic nature has many applications in applied mathematics,
statistics and physics \cite{zhang_schur_2005}. Here we use it to
``pull'' the coherence dynamic back into population space. Only
writing the Schur complement and skipping the other blocks of the
resolvent we have
\[
(z-M)^{-1}=\begin{pmatrix}(z-a^{\dagger}(b-z)^{-1}a)^{-1} & \cdot\\
\cdot & \cdot
\end{pmatrix}\,.
\]
Remember the operator $T$, the restriction to population space. With
our choice of density vector basis it has the form
\[
T=\begin{pmatrix}\one_{n} & 0_{n\times(n^{2}-n)}\end{pmatrix}\,.
\]
The difference of evolution for initial conditions $\vec{\rho}_{0}=\begin{pmatrix}\vec{p}_{0}\\
0
\end{pmatrix}=T^{\dagger}\vec{p}_{0}$ (zero coherences) between quantum network $M$ and kinetic network
$N$ is thus
\[
\left(Te^{Mt}T^{\dagger}-e^{Nt}\right)\vec{p}_{0}=\frac{1}{2\pi i}\oint e^{zt}\left((z-a^{\dagger}(b-z)^{-1}a)^{-1}-(z-N)^{\dagger}\right)\dz\,.
\]
For a good approximation we require
\begin{equation}
(z-a^{\dagger}(b-z)^{-1}a)^{-1}\approx(z-N)^{\dagger}\,.\label{eq:resolvent-approximation}
\end{equation}
At this point it is a small step to drop the second $z$ on the LHS,
in which case the formula becomes equality if we set
\[
N=a^{\dagger}b^{-1}a\,.
\]
To see intuitively that this approximation is good, consider the following.
Matrix $b$ contains terms proportional to $\sG$ on its diagonal
and terms proportional to $\sV$ on its off-diagonal, matrix $a$
is proportional to $\sV$, therefore
\[
N\propto\sV^{2}\sG^{-1}
\]
when $\sV\sG^{-1}$ becomes small. For values of $z$ that are smaller
than eigenvalues of $b$ the approximation (\ref{eq:resolvent-approximation})
is good because $(b-z)^{-1}\approx b^{-1}$, for values of $z$ larger
than eigenvalues of $b$ is good, because then $z$ is much larger
than the eigenvalues of $N$, and so both sides of (\ref{eq:resolvent-approximation})
are approximately $z^{-1}$. This basic insight is what drives our
bounds in Section~\ref{sec:Resolvent-difference-bounds}.

\subsection{Expanding $N$\label{sub:Expanding-N}}

As mentioned in \ref{sub:M}, $b$ consists of $2\times2$-blocks
proportional to $\sG$ on the diagonal and $2\times2$-blocks proportional
to $\sV$ on the off-diagonal. We separate this contributions defining
\[
b=b_{0}+\nu
\]
where $b_{0}\propto\sG$ and $\nu\propto\sV$ is the block-diagonal
and block-off-diagonal of $b$ respectively. If $\sV\sG^{-1}$ is
small enough and if $b_{0}$ is invertible we can expand
\[
b^{-1}=\sum_{k=0}^{\infty}b_{0}^{-1}\left(-\nu b_{0}^{-1}\right)^{k}\,.
\]
This leads to the expansion
\[
N=a^{\dagger}b^{-1}a=\sum_{k=0}^{\infty}N_{k}
\]
with
\begin{equation}
N_{k}=a^{\dagger}b_{0}^{-1}\left(-\nu b_{0}^{-1}\right)^{k}a\,.\label{eq:Nk-definition}
\end{equation}
When using explicit forms of $a$, $b_{0}$ and $\nu$ one can see
that the rates in $N_{k}$ consist of corrections due to interactions
via $k$ intermediates. Roughly speaking, every of the $(k+1)$ sites
along the chain contributes a factor of $\sV$, every of the $k$
coherences (links) contributes a factor of $\sG^{-1}$, thus $N_{k}$
scales like $\sV^{k+1}\sG^{-k}$.

\subsection{The network $N_{0}$\label{sub:kinetic-network-N0}}

We now present the explicit form of 
\[
N_{0}=a^{\dagger}b_{0}^{-1}a
\]
the dominant contribution to $N$. We only show the crucial parts
of the calculations that should make clear how to get the result for
general $n$.

Notice that, from \ref{sub:Expanding-N} and (\ref{eq:b-diagonal-blocks}),
it follows that $b_{0}$ is a $(n^{2}-n)\times(n^{2}-n)$ matrix with
the only nonzero entries being $2\times2$ blocks
\[
\begin{pmatrix}-\gamma_{kl} & -E_{kl}\\
E_{kl} & -\gamma_{kl}
\end{pmatrix}
\]
along the diagonal. With the unitary transformation
\[
U_{0}=\frac{1}{\sqrt{2}}\begin{pmatrix}-i & i\\
1 & 1
\end{pmatrix}
\]
we can diagonalize these $2\times2$ blocks. Hence, the entire matrix
$b_{0}$ can be diagonalized by applying the transformation
\begin{equation}
U=\one_{(n^{2}-n)/2}\otimes U_{0}\label{eq:U-transformation}
\end{equation}
and
\begin{equation}
\tilde{b}_{0}=U^{\dagger}b_{0}U=\diag(\alpha_{12},\bar{\alpha}_{12},\alpha_{13},\bar{\alpha}_{13},\dots,\bar{\alpha}_{n-1,n})\label{eq:b0-twidel}
\end{equation}
with
\[
\alpha_{kl}=-\gamma_{kl}+iE_{kl}
\]
where $\diag$ denotes a diagonal matrix with given diagonal entries.
In fact, $U$ also helps to simplify $a$, consider the case $n=3$
\begin{equation}
\tilde{a}=U^{\dagger}a=\begin{pmatrix}\overline{V}_{12} & -\overline{V}_{12}\\
V_{12} & -V_{12}\\
\overline{V}_{13} &  & -\overline{V}_{13}\\
V_{13} &  & -V_{13}\\
 & \overline{V}_{23} & -\overline{V}_{23}\\
 & V_{23} & -V_{23}
\end{pmatrix}\,,\label{eq:a-twidel-form}
\end{equation}
and the same happens for $\tilde{\nu}=U^{\dagger}\nu U$ (derivation
in Appendix~\ref{app:Three-sites}). Notice that both $\tilde{b}_{0}$
and $\tilde{a}$ are complex matrices, still, we can use the transformed
matrices $\tilde{a}$, $\tilde{b}_{0}$, and $\tilde{\nu}$ when finding
explicit expressions for the real matrices $N_{k}$, because $U$
cancels out. For example
\begin{align*}
N_{0} & =a^{\dagger}b_{0}^{-1}a\\
 & =a^{\dagger}UU^{\dagger}b_{0}^{-1}UU^{\dagger}a\,.\\
 & =(U^{\dagger}a)^{\dagger}(U^{\dagger}b_{0}^{-1}U)^{-1}(U^{\dagger}a)\\
 & =\tilde{a}^{\dagger}\tilde{b}_{0}^{-1}\tilde{a}\,.
\end{align*}
In the case $n=3$ we get
\begin{equation}
N_{0}=\begin{pmatrix}-\mu_{12}-\mu_{13} & \mu_{12} & \mu_{13}\\
\mu_{12} & -\mu_{12}-\mu_{23} & \mu_{23}\\
\mu_{13} & \mu_{23} & -\mu_{13}-\mu_{23}
\end{pmatrix}\label{eq:N0-rate-form}
\end{equation}
with
\[
\mu_{kl}=\frac{2\left|V_{kl}\right|^{2}\gamma_{kl}}{\gamma_{kl}^{2}+E_{kl}^{2}}\,.
\]
The following simplified calculation illustrates how the rates $\mu_{kl}$
result from the matrix multiplication $\tilde{a}^{\dagger}\tilde{b}_{0}^{-1}\tilde{a}$
\begin{align*}
\begin{pmatrix}\bar{V}\\
V
\end{pmatrix}^{\dagger}\begin{pmatrix}\alpha^{-1} & 0\\
0 & \bar{\alpha}^{-1}
\end{pmatrix}\begin{pmatrix}-\bar{V}\\
-V
\end{pmatrix} & =-V\bar{V}(\alpha^{-1}+\bar{\alpha}^{-1})\\
 & =\frac{2\left|V\right|^{2}\gamma}{\gamma^{2}+E^{2}}\,.
\end{align*}
More generally for any $n$ we have
\begin{equation}
\left(N_{0}\right)_{kl}=\mu_{kl}\label{eq:N0-form-a}
\end{equation}
for $i\neq j$ and
\begin{equation}
\left(N_{0}\right)_{kk}=-\sum_{l\neq k}\mu_{kl}\label{eq:N0-form-b}
\end{equation}
This is just the network described in \cite{cao_optimization_2009}
and the introduction.

\subsection{Including re-emission, recombination and trapping}

The population decreasing effects of re-emission, recombination and
trapping can all be described by the rates $\kappa_{k}$ of the diagonal
anti-hermitian operator
\[
A=\frac{-i}{2}\sum_{k=1}^{n}\kappa_{k}\ket k\bra k
\]
included in our general master equation (\ref{eq:general-master-equation}).
The contribution to the rate of change $\dot{\rho}$ is easily calculated
\[
-i\{A,\rho\}=-\sum_{k,l}\frac{1}{2}(\kappa_{k}+\kappa_{l})\ket i\bra j\,,
\]
and $M$ becomes
\[
M=\begin{pmatrix}c_{1} & -a^{\dagger}\\
a & b+c_{2}
\end{pmatrix}
\]
with the new contributions
\[
c_{1}=-\diag(\kappa_{1},\kappa_{2},\dots,\kappa_{n})
\]
and
\[
c_{2}=-\diag(\kappa_{12},\kappa_{12},\kappa_{13},\kappa_{13},\dots,\kappa_{n-1,n})
\]
with $\kappa_{kl}=\frac{1}{2}(\kappa_{k}+\kappa_{l})$ the rate that
decreases the coherence of sites $k$ and $l$. With this the networks
become
\begin{align}
N_{0} & =a^{\dagger}(b_{0}+c_{2})^{-1}a+c_{1}\nonumber \\
N & =a^{\dagger}(b+c_{2})^{-1}a+c_{1}\nonumber \\
N_{k} & =a^{\dagger}(b_{0}+c_{2})^{-1}\left(-\nu(b_{0}+c_{2})^{-1}\right)^{k}a\label{eq:all-N-complete}
\end{align}
which also hold with the replacements $a\to\tilde{a}$, $b\to\tilde{b}$
and $\nu\to\tilde{\nu}$, while leaving $c_{1}$ and $c_{2}$ unchanged.

The rates in $N_{0}$ can again be calculated directly
\[
\left(N_{0}\right)_{kl}=\mu_{kl}=\frac{2\left|V_{kl}\right|^{2}(\gamma_{kl}+\kappa_{kl})}{(\gamma_{kl}+\kappa_{kl})^{2}+E_{kl}^{2}}
\]
for $k\neq l$ and
\[
\left(N_{0}\right)_{kk}=-\kappa_{k}-\sum_{l\neq k}\mu_{kl}\,.
\]

\subsection{Numerical simulations\label{sub:Numerical-simulations}}

According to the last two subsections, network $N_{0}$ is easy to
calculate directly, while network $N$ and any $k$-site contribution
$N_{k}$ can be formed from the general definition of $\tilde{a}$,
$\tilde{b}_{0}$ and $\tilde{\nu}$ (see Appendix~(\ref{app:Constructing-nu}))
which can be somewhat tedious. However, there is another approach
related to numerical simulations. When running numerical calculations
to simulate a complex master equation (\ref{eq:general-master-equation})
on a software like Octave or Matlab, the need to convert the equation
to the form
\[
\dot{\vec{\rho}}=M\vec{\rho}
\]
with a real $M$ arises in any case. This can be done as we describe
it in \ref{sub:Converting-the-master}, or more easily -- because
we have the help of a computer -- by defining an orthonormal density
space basis. For example set
\begin{align*}
\sigma_{k} & =\ket k\bra k\\
\xi_{kl} & =(\ket k\bra l+\ket l\bra k)/\sqrt{2}\\
\eta_{kl} & =(-i\ket k\bra l+i\ket l\bra k)/\sqrt{2}
\end{align*}
for $k<l$. Then matrix $M$ can be formed by applying the master
equation to those vectors and finding their coordinates. The following
gives the population space block of $M$
\[
M_{kl}=\Tr{\sigma_{k}^{\dagger}\mathcal{M}(\sigma_{l})}
\]
where $\mathcal{M}$ is the superoperator formed by the RHS of the
master equation (\ref{eq:general-master-equation}). Once the entire
real matrix is found it is cut into population and coherence blocks
\[
M=\begin{pmatrix}m_{PP} & m_{PC}\\
m_{CP} & m_{CC}
\end{pmatrix}
\]
and a generalized kinetic network of the same form as $N$ is calculated
as
\[
N=m_{PC}m_{CC}^{-1}m_{CP}+m_{PP}\,.
\]
Hence, if one has already calculated $M$ in order to simulate a quantum
network, it only takes a few steps to find the kinetic network approximation
$N$.

\section{Preliminaries}

In this section we give some definitions and conditions. The conditions
allow us to infer basic facts about the spectra of the operators $N_{0}$,
$N$ and $M$, which are required for all our bounds in Sections~(\ref{sec:Bounding-relaxation-time}),
(\ref{sec:Bounding-evolution-error}) and (\ref{sec:Resolvent-difference-bounds}).

\subsection{Norm}

Because for our bounds of the relaxation time we remove all population
decreasing effects all evolutions $M$, $N_{0}$, and $N$ leave the
total population invariant. Therefore we split up the space of populations
$P$. Set
\[
\vec{e}=(1,1,\dots,1)^{\dagger}/n\in P
\]
the equal population vector. As we will prove in Proposition~\ref{prop:spectral-properties},
as long as the network meets certain conditions, both quantum and
kinetic evolutions will tend to $\vec{e}$ for any initial condition
with total population 1. Consequently, we are only interested in the
properties of our matrices in the space of population inequalities
\[
I=\vec{e}^{\top}=\left\{ \vec{v}\middle|\sum_{k}v_{k}=0\right\} \,.
\]

This is reflected in the norm we use, defines as follows. For $A:X_{1}\to X_{2}$
where $X_{1}$ and $X_{2}$ are equal to $I$ or $C$ we define the
operator norm as 
\[
\left\Vert A\right\Vert =\sup_{v\in X_{1}}\frac{\left\Vert Av\right\Vert _{2}}{\left\Vert v\right\Vert _{2}}
\]
where $\left\Vert \cdot\right\Vert _{2}$ is the Euclidean norm. Hence,
from now on, we think of our matrix blocks as
\begin{align*}
a: & I\to C\\
b: & C\to C\\
a^{\dagger}: & C\to I\,.
\end{align*}
Note that $\left\Vert a\right\Vert $ is the same if we maximize over
$I$ or $P$ because $a\vec{e}=0$, and that $\left\Vert a\right\Vert =\left\Vert a^{\dagger}\right\Vert $.
Also, according to Proposition~\ref{prop:spectral-properties} $N_{0}<0$
on $I$ and therefore $N_{0}^{-1}$ is well-defined. The same holds
for $N$. Define
\[
\mu=\left\Vert N^{-1}\right\Vert ^{-1}
\]
to be the eigenvalue closest to 0 in $N$ in $I$, define $\mu_{0}$
the same way for $N_{0}$.

\subsection{Conditions\label{sub:Conditions}}

For all our following bounds we have a set of conditions.
\begin{itemize}
\item First, we require that the network is \emph{connected}, in the sense
that any two sites $k$ and $l$ are coupled, at least via some intermediates,
i.e. for some integer $p\geq0$ there are sites $m_{j}$, $j=1\dots p$
such that the product
\[
V_{km_{1}}V_{m_{1}m_{2}}\dots V_{m_{p-1}m_{p}}V_{k_{p}j}
\]
is nonzero. This condition ensures that all sites can exchange population
and the evolution ultimately converges to $\vec{e}$.
\item Second, we require all the site dephasing rates to be strictly positive,$\gamma_{k}>0$.
This condition is essential for our approximation, as the coherences
need to decay for the evolution $M$ to become non-oscillatory. Notice
that the limit $\sV\sG^{-1}\to0$ does not require that the dephasing
rates get larger, but they will be much larger than the magnitude
of eigenvalues of $N$ or $N_{0}$, the population decay rates, because
$\sG\gg\sV^{2}\sG^{-1}$. 
\item Finally, we require that the $V_{kl}$ are \emph{real}. This ensures
that $N$ is symmetric and has a real spectrum (see Proposition~\ref{prop:spectral-properties}),
which allows simpler bounds in our proofs. While $N_{0}$ is always
symmetric, we first compare the evolutions of $M$ and $N$, and then
the evolutions of $N$ and $N_{0}$. Therefore we require this condition
for both $N_{0}$ and $N$. We are confident that our methods would
extend to the case of complex $V_{kl}$, but for the sake of clarity
we restrict ourselves to the simpler case.
\end{itemize}

\subsection{Inverse bounds\label{sub:Inverse-bounds}}

Or proofs consist mainly of using the following two bounds on the
inverse on different parts of resolvents.

First, consider the Taylor series of the inverse close to 1, which
for real numbers $x$ gives
\[
|(1+x)^{-1}-1|\leq2|x|
\]
for $|x|\leq1/2$. This is readily translated to a bound for operators
\begin{equation}
\left\Vert (A+B)^{-1}-A^{-1}\right\Vert \leq2\left\Vert A^{-1}\right\Vert ^{2}\left\Vert B\right\Vert \label{eq:inverse-bound}
\end{equation}
for $\left\Vert B\right\Vert \leq\frac{1}{2}\left\Vert A^{-1}\right\Vert ^{-1}$.

Second, if $A<-c<0$ is a negative definite, self-adjoint, finite
dimensional operator and $z\in\C$ with $\re z\ge0$ then
\begin{equation}
\left\Vert (z-A)^{-1}\right\Vert \leq c^{-1}\label{eq:inv-bound-const}
\end{equation}
and
\begin{equation}
\left\Vert (z-A)^{-1}\right\Vert \leq\left|z\right|^{-1}\,.\label{eq:inv-bound-yinv}
\end{equation}
these two bounds follow from the fact
\begin{align*}
\left\Vert (z-A)^{-1}\right\Vert  & =\max\,\left\{ \left|\lambda\right|\middle|\lambda\in\spec(z-A)^{-1}\right\} \\
 & =\max\,\left\{ \left|(z-\lambda)^{-1}\right|\middle|\lambda\in\spec A\right\} \,.
\end{align*}

\subsection{Spectral properties}

The following Proposition gives some basic facts about the spectra
of the kinetic networks $N$ and $N_{0}$. We will use these properties
for the proofs of our bounds.
\begin{prop}
\label{prop:spectral-properties}The matrices $N_{0}$ and $N$ as
defined in \ref{sub:kinetic-network-N0} and \ref{sub:kinetic-network-N}
have the following properties
\begin{enumerate}
\item $N_{0}$ is real and symmetric.
\item $N$ is real.
\item If the interactions $V_{kl}$ are real then $N$ is symmetric.
\item $N_{0}\vec{e}=N\vec{e}=0$
\item If $\gamma_{k}>0$ and the network is connected network then $N_{0}<0$
on $I$.
\item For $\sV\sG^{-1}$ small enough, $N_{0}<-\mu/2$ and $N<-\mu_{0}/2$
on $I$.
\end{enumerate}
\end{prop}
\begin{proof}
1. These properties follow directly from the form in (\ref{eq:N0-form-a})
and (\ref{eq:N0-form-b}).

2. $N$ is real because it is a product of $a$, $b^{-1}$ and $a^{\dagger}$
which are also real.

3. If $V_{kl}$ is real then $\tilde{a}$ (see (\ref{eq:a-twidel-form}))is
real, so
\[
N=\tilde{a}^{\top}\tilde{b}^{-1}\tilde{a}
\]
Furthermore, $\tilde{b}^{\top}=\tilde{b}$ (see Appendix~\ref{app:Three-sites})
therefore
\begin{align*}
N^{\top} & =\tilde{a}^{\top}\left(\tilde{b}^{-1}\right)^{\top}\tilde{a}\\
 & =\tilde{a}^{\top}\left(\tilde{b}^{\top}\right)^{-1}\tilde{a}\\
 & =N\,.
\end{align*}

4. From (\ref{eq:a-twidel-form}) it is not hard to understand how
$\tilde{a}$ looks for any $n$. One sees that the two rows for the
coherence between sites $k$ and $l$ have exactly two non-zero entries,
the first has $V_{kl}$ and $-V_{kl}$, and the second has $\bar{V}_{kl}$
and $-\bar{V}_{kl}$. Therefore $\tilde{a}\vec{e}=0$ and so $N_{0}\vec{e}=N\vec{e}=0$.

5. For $\vec{v}\in I$ we have $\sum_{k}v_{k}=0$. Now 
\begin{align*}
\vec{v}^{\dagger}N_{0}\vec{v} & =-\sum_{k<l}\mu_{kl}(v_{k}-v_{l})^{2}\\
 & \leq0
\end{align*}
The condition for equality is as follows. Because $\gamma_{k}>0$
we have
\begin{align*}
V_{kl}\neq0 & \iff\mu_{kl}\neq0\,.
\end{align*}
Hence, because the network is connected, we have $v_{k}=v_{l}$ for
all $k$ and $l$ and with $\sum_{k}v_{k}=0$ it follows that $\vec{v}=0$,
thus $N_{0}<0$ on $I$.

6. Because $\mu_{0}=\left\Vert N_{0}^{-1}\right\Vert ^{-1}$ and $N_{0}<0$
we have $N_{0}\leq-\mu_{0}$ on $I$. Note that $\mu\propto\sV^{2}\sG^{-1}$
can grow as $\sV\sG^{-1}$ gets small, so $\mu-\mu_{0}$ can grow
in absolute value, however, as we now show the spectra of $N$ and
$N_{0}$ approach each other relative to their ``size''
\[
\mu-\mu_{0}\ll\mu_{0}
\]
We bound the distance of $N$ and $N_{0}$ with the inverse bound.
For $\sV\sG^{-1}$ small enough we have $\left\Vert \nu\right\Vert \leq\frac{1}{2}\left\Vert b_{0}^{-1}\right\Vert ^{-1}$
and we can apply (\ref{eq:inverse-bound}) on
\[
\left\Vert (b_{0}+\nu)^{-1}-b_{0}^{-1}\right\Vert \leq2\left\Vert b_{0}^{-1}\right\Vert ^{2}\left\Vert \nu\right\Vert \,.
\]
Now
\begin{align*}
\left\Vert N-N_{0}\right\Vert  & =\left\Vert a^{\dagger}\left(b^{-1}-b_{0}^{-1}\right)a\right\Vert \\
 & =\left\Vert a\right\Vert ^{2}2\left\Vert b_{0}^{-1}\right\Vert ^{2}\left\Vert \nu\right\Vert \,.
\end{align*}
So, the distance of $N$ and $N_{0}$ is proportional to $\sV^{2}\sG^{-2}\sV$,
and the eigenvalues in $N_{0}$ and $N$ -- in particular $\mu_{0}$
and $\mu$ -- are proportional to $\sV^{2}\sG^{-1}$. Comparing the
two gives
\[
\sV^{2}\sG^{-2}\sV\ll\sV^{2}\sG^{-1}
\]
because $\sV\sG^{-1}\ll1$. That means the eigenvalues are approaching
each other relative to their magnitude, in particular $N$ becomes
negative definite like $N_{0}$, and
\[
\frac{\left|\mu-\mu_{0}\right|}{\mu_{0}}\to0\,.
\]
Now it immediately follows that
\begin{align*}
N<- & \mu<\mu_{0}/2\\
N_{0}<- & \mu_{0}<\mu/2\,.
\end{align*}
\end{proof}

\section{Bounding relaxation time error\label{sec:Bounding-relaxation-time}}

We now give an explicit definition of relaxation time and the norms
we use to control it. Then we derive bounds first comparing the quantum
network $M$ to the kinetic network $N$, and then comparing the kinetic
networks $N$ and $N_{0}$.

As a simple check of sanity consider the following. If we scale $\sG\propto s$
and $\sV\propto s$ then also $M,N\propto s$ and time scales inversely
$\Delta\tau,\tau\propto s^{-1}$. Therefore the relative error $\Delta\tau_{\mathrm{rel}}=\Delta\tau/\tau$
stays unchanged and we expect bounds in terms of positive powers of
$\sV\sG^{-1}$. Our two bounds show exactly this behavior. The approximation
of $N$ to $M$ is proportional to $\sV^{2}\sG^{-2}$, while the approximation
of $N_{0}$ to $N$ is proportional to $\sV\sG^{-1}$, combining the
two approximations it follows that the approximation of $N_{0}$ to
$M$ is also proportional $\sV\sG^{-1}$.

Note that all the results in this Section require the conditions in
\ref{sub:Conditions}.

\subsection{Relaxation time}

By Proposition~(\ref{prop:spectral-properties}), the eigenvalues
of $N$ and $N_{0}$ on $I$ are all negative for $\sV\sG^{-1}$ small
enough, so for any initial distribution $\vec{p}_{0}\in I$
\begin{align*}
e^{Nt}\vec{p}_{0} & \to0\\
e^{N_{0}t}\vec{p}_{0} & \to0
\end{align*}
for large $t$. We can integrate
\begin{align*}
\int_{0}^{\infty}e^{Nt}\dt & =N^{-1}
\end{align*}
and applying the operator norm maximizes the relaxation time for the
kinetic network $N$ over all possible population inequalities $\vec{p}_{0}\in I$,
set
\begin{align*}
\tau=\mu^{-1} & =\no{N^{-1}}=\left\Vert (a^{\dagger}b^{-1}a)^{-1}\right\Vert \\
 & =\left\Vert \int_{0}^{\infty}e^{Nt}\dt\right\Vert 
\end{align*}
and in the same way we define $\tau_{0}=\mu_{0}^{-1}$ for the network
$N_{0}$. 

We define the error in relaxation time as the relaxation time difference
maximized over $I$
\[
\Delta\tau=\no{\int_{0}^{\infty}Te^{Mt}T^{\dagger}-e^{Nt}\dt}\,.
\]
Hence, bounding $\Delta\tau$ means controlling\emph{ }the \emph{worst
possible} error in relaxation time when approximating $M$ by $N$.
The relative error is
\[
\Delta\tau_{\mathrm{rel}}=\Delta\tau/\tau
\]
notice that we compare the worst possible relaxation time error to
the longest possible relaxation time, those two do not necessarily
occur for the same initial condition. We define $\Delta\tau_{0}$
and $\Delta\tau_{0,\,\mathrm{rel}}$ in the same way, comparing $N$
and $N_{0}$.

\subsection{Resolvent difference\label{sub:rewrite-for-Resolvent-difference}}

Converting the operator for the relaxation time error we get
\begin{align*}
\int_{0}^{\infty}Te^{Mt}T^{\dagger}-e^{Nt}\dt & =TM^{-1}T^{\dagger}-N^{-1}\\
 & =\frac{1}{2\pi i}\oint\frac{1}{z}\left(T\frac{1}{z-M}T^{\dagger}-\frac{1}{z-N}\right)\dz\\
 & =\frac{1}{2\pi i}\oint\frac{1}{z}\left(\frac{1}{z-a^{\dagger}(b-z)^{-1}a}-\frac{1}{z-a^{\dagger}b^{-1}a}\right)\dz\,,
\end{align*}
where the complex integration follows a contour surrounding both $\spec M$
and $\spec N$. Define $S(z)$ to be the difference of the two resolvents
\[
S(z)=\frac{1}{z-a^{\dagger}(b-z)^{-1}a}-\frac{1}{z-a^{\dagger}b^{-1}a}\,.
\]
We now seek a bound on
\begin{equation}
\no{\int_{0}^{\infty}Te^{Mt}T^{\dagger}-e^{Nt}\dt}=\left\Vert \frac{1}{2\pi i}\oint\frac{1}{z}S(z)\dz\right\Vert \,.\label{eq:delta-t}
\end{equation}

\subsection{Comparing the relaxation time of $M$ and $N$}

When bounding second order terms with the inverse bound we encounter
\begin{align*}
\kappa & =\left\Vert a\right\Vert ^{2}\left\Vert b^{-1}\right\Vert ^{2}
\end{align*}
and $\kappa_{0}$ for the corresponding terms with $b_{0}$ instead
of $b$. Notice the scaling behavior $\mu,\mu_{0}\propto\sV^{2}\sG^{-1}$
and $\kappa,\kappa_{0}\propto\sV^{2}\sG^{-2}$.

We will change the contour integration in (\ref{eq:delta-t}) to be
along the imaginary axis $z=iy$ for $y\in\R$. We prove the somewhat
technical bounds on $S(iy)$ in Lemma~\ref{lem:decay-bounds-right-plane}
in Section~\ref{sec:Resolvent-difference-bounds}.
\begin{thm}
\label{thm:relax-bound-1}If $\sV\sG^{-1}$ is small enough then
\[
\Delta\tau=\no{\int_{0}^{\infty}Te^{Mt}T^{\dagger}-e^{Nt}\dt}\leq\frac{4}{\pi}\kappa\mu^{-1}(1+\beta)
\]
where $\beta>0$ is the scaling independent constant from Lemma~\ref{lem:decay-bounds-right-plane}.
This gives a bound on the relative error
\[
\Delta\tau_{\mathrm{rel}}=\Delta\tau/\tau\leq\frac{4}{\pi}\kappa(1+\beta)=k_{2}\sV^{2}\sG^{-2}
\]
where $k_{2}$ is scaling invariant.\end{thm}
\begin{proof}
We set the integration contour in (\ref{eq:delta-t}) to be along
the complex axis $z=iy$ for $y\in\R$ with $y$ going from $-R$
to $+R$. We close the contour to the left in the half plane of negative
real parts along a circle of radius $R$. According to Lemma~\ref{lem:decay-bounds-right-plane},
$S(z)$ has no poles with $\re z\geq0$ and so all poles lie within
this contour for $R$ large enough and $\sV\sG^{-1}$ small enough.
As $R$ tends to infinity the integrand behaves like $\frac{1}{z^{3}}$
so the half-circle does not contribute to the integral. We can therefore
change to complex integral to an integral in $y$ over all of $\R$
\[
\no{\int_{0}^{\infty}Te^{Mt}T^{\dagger}-e^{Nt}\dt}=\left\Vert \frac{1}{2\pi}\int_{\R}\frac{1}{iy}S(iy)\dy\right\Vert \,.
\]

Now split up the integral into two regions $|y|\leq\mu$ and $|y|\geq\mu$
and then use the corresponding bounds from Lemma~\ref{lem:decay-bounds-right-plane}.
Choose $\sV\sG^{-1}$ small enough so that $\mu<\alpha$ and use part
1 of the Lemma to bound
\begin{align*}
\left\Vert \int_{-\mu}^{\mu}\frac{1}{iy}S(iy)\dy\right\Vert  & \leq\int_{-\mu}^{\mu}\frac{1}{|y|}\cdot4\kappa\mu^{-2}|y|\dy\\
 & \leq8\kappa\mu^{-1}\,,
\end{align*}
and use part 2 of the Lemma to bound
\begin{align*}
\left\Vert \int_{\mu}^{\infty}\frac{1}{iy}S(iy)\dy\right\Vert  & \leq\int_{\mu}^{\infty}\frac{1}{|y|}\cdot4\beta\kappa|y|^{-1}\dy\\
 & \leq4\beta\kappa\mu^{-1}\,.
\end{align*}
Adding the two bounds gives the result
\begin{align*}
\no{\int_{0}^{\infty}Te^{Mt}T^{\dagger}-e^{Nt}\dt} & \leq\left\Vert \frac{1}{2\pi i}\int_{\R}\frac{1}{iy}S(iy)\dy\right\Vert \\
 & \leq\frac{1}{2\pi}8\kappa\mu^{-1}(1+\beta)\,.
\end{align*}

\end{proof}

\subsection{Comparing the relaxation time of $N$ and $N_{0}$}
\begin{thm}
\label{thm:relax-bound-2}If $\sV\sG^{-1}$ is small enough then
\[
\Delta\tau_{1}=\no{\int_{0}^{\infty}e^{Nt}-e^{N_{0}t}\dt}\leq4\kappa\mu^{-2}\left\Vert \nu\right\Vert 
\]
where $\mu$ and $\kappa$ can also be replaced by $\mu_{0}$ and
$\kappa_{0}$. This gives a bound on the relative error
\[
\Delta\tau_{1,\,\mathrm{rel}}=\Delta\tau_{1}/\tau\leq4\kappa\mu^{-1}\left\Vert \nu\right\Vert =k'_{2}\sV\sG^{-1}
\]
where $k_{2}'$ is scaling invariant.\end{thm}
\begin{proof}
In this case we don't need to bound the resolvent, instead we can
evaluate the integral
\[
\int_{0}^{\infty}e^{Nt}-e^{N_{0}t}\dt=N^{-1}-N_{0}^{-1}\,.
\]
We use the inverse bound (\ref{eq:inverse-bound}) twice. First, because
$\left\Vert \nu\right\Vert \leq\frac{1}{2}\left\Vert b^{-1}\right\Vert ^{-1}$
as long as $\sV\sG^{-1}$ is small enough, we can apply the bound
on
\begin{equation}
\left\Vert (b-\nu)^{-1}-b^{-1}\right\Vert \leq2\left\Vert b^{-1}\right\Vert ^{2}\left\Vert \nu\right\Vert \,.\label{eq:dec-time-inside-bound}
\end{equation}
Now, apply the bound again with $A=N$ and $B=N_{0}-N$. The condition
for $B$ is 
\begin{align*}
\left\Vert B\right\Vert  & \leq\left\Vert a\right\Vert ^{2}\left\Vert (b-\nu)^{-1}-b^{-1}\right\Vert \\
 & \leq\left\Vert a\right\Vert ^{2}2\left\Vert b^{-1}\right\Vert ^{2}\left\Vert \nu\right\Vert \\
 & =2\kappa\left\Vert \nu\right\Vert \\
 & \leq2\left\Vert A^{-1}\right\Vert ^{-1}=2\mu
\end{align*}
where we used (\ref{eq:dec-time-inside-bound}) in the second step.
The last inequality is again achieved for $\sV\sG^{-1}$ small enough
because the two sides scale like
\[
\sV^{2}\sG^{-2}\sV\leq\sV^{2}\sG^{-1}\,.
\]
Now it follows that
\[
\left\Vert N^{-1}-N_{0}^{-1}\right\Vert \leq2\left\Vert A^{-1}\right\Vert ^{2}\left\Vert B\right\Vert =4\kappa\mu^{-2}\left\Vert \nu\right\Vert 
\]
as claimed. By switching the role of $b$ and $b_{0}$ we receive
the corresponding bound with $\kappa_{0}$ and $\mu_{0}$.
\end{proof}
As a corollary we receive a bound on the relaxation time difference
between the fully quantum mechanical evolution of $M$ and the simple
kinetic network evolution of $N_{0}$.
\begin{cor}
\label{cor:relax-bound}If $\sV\sG^{-1}$ is small enough then for
some scaling independent constant $k_{1}$
\[
\Delta\tau_{0,\,\mathrm{rel}}\leq k_{1}\sV\sG^{-1}
\]
\end{cor}
\begin{proof}
According to Proposition~\ref{prop:spectral-properties} we have
$\left|\mu-\mu_{0}\right|/\mu_{0}\to0$, and therefore there is a
$c$ such that
\[
c\geq\tau/\tau_{0}
\]
for $\sV\sG^{-1}$ small enough. Then with Theorems~\ref{thm:relax-bound-1}
and \ref{thm:relax-bound-2} we have
\begin{align*}
\Delta\tau_{0,\,\mathrm{rel}} & =\Delta\tau_{0}/\tau_{0}\\
 & \leq(\Delta\tau+\Delta\tau_{1})/\tau_{0}\\
 & \leq c(\Delta\tau+\Delta\tau_{1})/\tau\\
 & \leq c\left(\Delta\tau_{\mathrm{rel}}+\Delta\tau_{1,\,\mathrm{rel}}\right)\\
 & \leq k_{1}\sV\sG^{-1}\,.
\end{align*}
\end{proof}

\section{Bounding evolution error\label{sec:Bounding-evolution-error}}

In this chapter we bound the difference of time evolution operators
for $M$, $N$ and $N_{0}$. Our error bounds looks as follows
\[
\no{e^{Mt}-e^{Nt}}\leq e^{-\mu t/2}\cdot X
\]
where $X$ is proportional to $\sV^{2}\sG^{-2}$ up to a logarithmic
term, and proportional to $\sV\sG^{-1}$ if $N$ is replaced with
$N_{0}$. The logarithmic term appears due to intermediate times.
It seems the integral over time performed in the last chapter seems
to have conveniently guided us around that logarithm. As for the time
dependence, using a shifting integration contour might give a bound
like $e^{-\mu t}\mu t$, but a better control of the spectrum would
be necessary to shift the contour close to $-\mu$ for long times.

As in the last chapter in \ref{sub:rewrite-for-Resolvent-difference},
we write the evolution difference as a complex integral before we
prove bounds
\begin{equation}
\no{Te^{Mt}T^{\dagger}-e^{Nt}}=\left\Vert \frac{1}{2\pi i}\oint e^{zt}S(z)dt\right\Vert \,.\label{eq:delta-t-1}
\end{equation}
Note that all the results in this Section again require the conditions
in \ref{sub:Conditions}.

\subsection{Comparing the evolution of $M$ and $N$}

We will change the contour integration in (\ref{eq:delta-t-1}) to
be parallel to the imaginary axis $z=iy-\mu/2$ for $y\in\R$. With
this choice the exponential in the integral yields exponential decay
at rate $\mu/2$. Again we give the technical bounds on $S(iy-\mu/2)$
in Lemma~\ref{lem:decay-bounds-sliver} in Section~\ref{sec:Resolvent-difference-bounds}.
\begin{thm}
\label{thm:evolution-bound-1}If $\sV\sG^{-1}$ is small enough then
for all $t\geq0$ we have
\[
\no{Te^{Mt}T^{\dagger}-e^{Nt}}\leq e^{-\mu t}\cdot k_{4}\sV^{2}\sG^{-2}\left(1+k_{5}\ln\sV^{-1}\sG\right)
\]
where $k_{4}$ and $k_{5}$ are a scaling independent constants.\end{thm}
\begin{proof}
We set the integration contour in (\ref{eq:delta-t-1}) to parallel
to the complex axis $z=iy-\mu/2$ for $y\in\R$ with $y$ going from
$-R$ to $+R$. We close the contour to the left in the half plane
of negative real parts along a circle or radius $R$. According to
Lemmas~\ref{lem:decay-bounds-right-plane} and \ref{lem:decay-bounds-sliver},
$S(z)$ is bounded for $\re z\geq-\mu/2$ and hence has no poles.
Therefore all the poles lie within the contour for $R$ large enough
and $\sV\sG^{-1}$ small enough. As $R$ tends to infinity the integrand
behaves like $\frac{1}{z^{2}}e^{\re z}$ so the half-circle does not
contribute to the integral. We can therefore change to complex integral
to an integral in $y$ over all of $\R$
\[
\no{Te^{Mt}T^{\dagger}-e^{Nt}}=\left\Vert \frac{1}{2\pi}\int_{\R}e^{(iy-\mu/2)t}S(iy-\mu/2)\dy\right\Vert \,.
\]

Now split up the integral into three regions with $|y|$ in the intervals
$[0,\mu]$, $[\mu,\hat{\alpha}]$ and $[\hat{\alpha},+\infty)$ and
then use the bounds from Lemma~\ref{lem:decay-bounds-sliver}. Choose
$\sV\sG^{-1}$ small enough so that $\mu<\hat{\alpha}$ and use part
1 of the Lemma to bound
\begin{align*}
\left\Vert \int_{0}^{\mu/2}e^{(iy-\mu/2)t}S(iy-\mu/2)\dy\right\Vert  & \leq e^{-\mu t/2}\int_{0}^{\mu/2}16\kappa\mu^{-2}|iy-\mu/2|\dy\\
 & \leq e^{-\mu t/2}16\kappa
\end{align*}
and
\begin{align*}
\left\Vert \int_{\mu/2}^{\hat{\alpha}}e^{(iy-\mu/2)t}S(iy-\mu/2)\dy\right\Vert  & \leq e^{-\mu t/2}\int_{\mu/2}^{\hat{\alpha}}4\kappa\cdot|y|^{-2}|iy-\mu/2|\dy\\
 & \leq e^{-\mu t/2}\int_{\mu/2}^{\hat{\alpha}}4\kappa\cdot2y^{-1}\dy\\
 & =e^{-\mu t/2}8\kappa\ln(2\hat{\alpha}/\mu)\,,
\end{align*}
and we use part 2 of the Lemma to bound
\begin{align*}
\left\Vert \int_{\hat{\alpha}}^{\infty}e^{(iy-\mu/2)t}S(iy-\mu/2)\dy\right\Vert  & \leq e^{-\mu t/2}\int_{\hat{\alpha}}^{\infty}4|y|^{-2}\no a^{2}(b_{\min}-\mu/2)^{-1}\dy\\
 & \leq e^{-\mu t/2}4\hat{\alpha}^{-1}\no a^{2}(b_{\min}-\mu/2)^{-1}\,.
\end{align*}
Adding the three bounds gives the result
\begin{align*}
\no{Te^{Mt}T^{\dagger}-e^{Nt}} & \leq e^{-\mu t/2}4\left(4\kappa+2\kappa\ln(2\hat{\alpha}/\mu)+\hat{\alpha}^{-1}\no a^{2}(b_{\min}-\mu/2)^{-1}\right)\,.
\end{align*}
The middle term of the parenthesis has the worst scaling behavior
\[
2\kappa\ln(2\hat{\alpha}/\mu)\propto\sV^{2}\sG^{-2}\ln\sV^{-1}\sG
\]
while the other two terms scale like $\sV^{2}\sG^{-2}$. Therefore
there are some scaling independent constants $k_{4}$ and $k_{5}$
such that
\[
\no{Te^{Mt}T^{\dagger}-e^{Nt}}\leq e^{-\mu t/2}\cdot k_{4}\sV^{2}\sG^{-2}\left(1+k_{5}\ln\sV^{-1}\sG\right)\,.
\]

\end{proof}

\subsection{Comparing the evolution of $N$ and $N_{0}$}
\begin{thm}
\label{thm:evolution-bound-2}If $\sV\sG^{-1}$ is small enough then
for all $t\geq0$ we have
\[
\no{e^{Nt}-e^{N_{0}t}}\leq e^{-\mu t/2}\cdot k_{4}'\sV\sG^{-1}
\]
where $k_{4}'$ is scaling independent, and where $\mu$ and $\kappa$
can also be replaced by $\mu_{0}$ and $\kappa_{0}$.\end{thm}
\begin{proof}
We are bounding the integral
\[
\frac{1}{2\pi i}\oint e^{zt}\tilde{S}(z)dz
\]
with resolvent difference
\[
\tilde{S}(z)=\frac{1}{z-N}-\frac{1}{z-N_{0}}\,.
\]
We use the same contour as in Proposition~\ref{thm:evolution-bound-1},
$z=iy-\mu/2$. According to Proposition~\ref{prop:spectral-properties},
all poles of $\tilde{S}(z)$ lie within this contour when $\sV\sG^{-1}$
is small enough and $R$ is large enough. Because of the $e^{zt}$
factor and $T(z)$ tending to zero, the integral over the half-circle
tends to 0 as $R$ becomes large.

We bound $\tilde{S}(z)$ in much the same way that we bounded $S(z)$
in Lemma~\ref{lem:decay-bounds-sliver}, however, the procedure is
more straightforward. Set
\[
X=(b-\nu)^{-1}-b^{-1}\,,
\]
because $\left\Vert \nu\right\Vert \leq\frac{1}{2}\left\Vert b^{-1}\right\Vert ^{-1}$
we can use the inverse bound (\ref{eq:inverse-bound})
\[
\left\Vert X\right\Vert \leq2\left\Vert b^{-1}\right\Vert ^{2}\left\Vert \nu\right\Vert \,.
\]
Now rewrite
\[
\tilde{S}(z)=(z-a^{\dagger}b^{-1}a)^{-1}+(z-a^{\dagger}b^{-1}a-a^{\dagger}Xa)^{-1}\,.
\]
For any $z$ with $\re z=-\mu/2$ and for $\sV\sG^{-1}$ small enough
we have
\begin{align*}
\left\Vert a^{\dagger}Xa\right\Vert  & \leq2\left\Vert a\right\Vert ^{2}\left\Vert b^{-1}\right\Vert ^{2}\left\Vert \nu\right\Vert \\
 & \leq\frac{1}{2}\left\Vert z-a^{\dagger}b^{-1}a\right\Vert 
\end{align*}
and so we can apply (\ref{eq:inverse-bound}) again
\begin{align*}
\left\Vert \tilde{S}(z)\right\Vert  & \leq2\left\Vert (z-a^{\dagger}ba)^{-1}\right\Vert ^{2}\left\Vert a^{\dagger}Xa\right\Vert \\
 & \leq4\left\Vert (z-a^{\dagger}ba)^{-1}\right\Vert ^{2}\kappa\left\Vert \nu\right\Vert \,.
\end{align*}
Now we apply inverse bounds (\ref{eq:inv-bound-const}) and (\ref{eq:inv-bound-yinv})
to receive the bounds
\begin{align}
\left\Vert \tilde{S}(z)\right\Vert  & \leq16\mu^{-2}\kappa\left\Vert \nu\right\Vert \label{eq:tz-const-bound}\\
\left\Vert \tilde{S}(z)\right\Vert  & \leq4\left|z\right|^{-2}\kappa\left\Vert \nu\right\Vert \label{eq:tz-z-2-bound}
\end{align}
as long as $\re z=-\mu/2$.

To estimate the integral
\[
\frac{1}{2\pi}\int_{\R}e^{(iy-\mu/2)t}T(iy-\mu/2)\dy
\]
we split it into the two regions $\left|y\right|\leq\mu/2$ and $\left|y\right|>\mu/2$.
Use (\ref{eq:tz-const-bound}) to bound
\begin{align*}
\left\Vert \int_{-\mu/2}^{\mu/2}e^{(iy-\mu/2)t}\tilde{S}(iy-\mu/2)\dy\right\Vert  & \leq e^{-\mu t/2}\cdot\mu\cdot16\mu^{-2}\kappa\left\Vert \nu\right\Vert \,,
\end{align*}
and use (\ref{eq:tz-z-2-bound}) to bound
\begin{align*}
2\left\Vert \int_{\mu/2}^{\infty}e^{(iy-\mu/2)t}\tilde{S}(iy-\mu/2)\dy\right\Vert  & \leq2e^{-\mu t/2}\cdot8\mu^{-1}\kappa\left\Vert \nu\right\Vert \,.
\end{align*}
Adding the two bounds gives
\begin{align*}
\no{e^{Nt}-e^{N_{0}t}} & \leq32e^{-\mu t/2}\mu^{-1}\kappa\left\Vert \nu\right\Vert \\
 & \leq e^{-\mu t/2}\cdot k_{6}\sV\sG^{-1}\,.
\end{align*}
where $k_{6}$ is scaling independent. The whole proof works just
as well when exchanging $\mu$ with $\mu_{0}$, $\kappa$ with $\kappa_{0}$
giving a similar bound.
\end{proof}
As a corollary we receive a bound on the decay time difference between
the fully quantum mechanical evolution of $M$ and the simple kinetic
network evolution of $N_{0}$.
\begin{cor}
\label{cor:evolution-bound}If $\sV\sG^{-1}$ is small enough then
for all $t\geq0$ we have
\[
\no{Te^{Mt}T^{\dagger}-e^{N_{0}t}}\leq e^{-\mu t/2}\cdot k_{3}\sV\sG^{-1}
\]
where $k_{3}$ is a scaling independent constant.\end{cor}
\begin{proof}
The bound follows from Theorems~\ref{thm:evolution-bound-1} and
\ref{thm:evolution-bound-2} and the fact that
\[
\sV^{2}\sG^{-2}\ln\sV^{-1}\sG\leq\sV\sG^{-1}
\]
for $\sV\sG^{-1}\leq1$.\end{proof}

\section{Applications}

The rate of direct population exchange 
\[
\mu_{kl}=\frac{2\left|V_{kl}\right|^{2}\gamma_{kl}}{\gamma_{kl}^{2}+E_{kl}^{2}}
\]
determines the strength of the link between sites $k$ and $l$ for
the network $N_{0}$. Because of our condition that $\gamma_{k}>0$,
the network topology is fully determined by the $V_{kl}$, but the
strength of the links is also affected by $\gamma_{k}$ and $E_{k}$.

As applications, we consider two idealized networks. The first is
a highly connected network where all sites are linked, the second
is a circular chain where where only nearest neighbors are linked.
We numerically calculate the relaxation times for the networks $M$,
$N_{0}$ and $N$ and compare the relative errors. Then we compare
these networks to randomized networks with the same network topology.
We also discuss the dimension dependence of our bounds from Sections~\ref{sec:Bounding-relaxation-time}
and again compare it to numerical simulations. All the simulations
agree with our bounds, but they show much room for improvement when
considering large dimensions.

Finally, we discuss the FMO-complex and our model for which some results
were already shown in the introduction in Figure~\ref{fig:FMO-efficiency}.

For clarity of notation we recall that $\Delta\tau$, $\Delta\tau_{0}$
and $\Delta\tau_{1}$ are relaxation time differences between the
network pairs $M-N$, $M-N_{0}$ and $N-N_{0}$ respectively. This
only makes the discussion more precise, while generally $\Delta\tau_{0}$
and $\Delta\tau_{1}$ show the same dimension and scaling behavior,
with small corrections to constants.

\subsection{Highly connected network}

Consider a highly connected network
\begin{align*}
V_{kl} & =\sV\\
E_{k} & =0\\
\gamma_{k} & =\sG\,.
\end{align*}
In Figure~\ref{fig:highly-connected} we made a plot for the computed
relative relaxation time differences $\Delta\tau_{\mathrm{rel}}$
and $\Delta\tau_{0,\,\mathrm{rel}}$ for different $\sV\sG^{-1}$
with the initial state localized at site 1. Both axes plot logarithms,
hence a straight line with slope $n$ represents a $(\sV\sG^{-1})^{n}$
proportionality. 

The difference $\Delta\tau_{\mathrm{rel}}$ is too small to show any
clear behavior. The difference $\Delta\tau_{0,\,\mathrm{rel}}$ is
linear with slope approximately $2$, hence the approximation is better
than the slope $1$ expected from Theorem~\ref{thm:relax-bound-2}.
In the same figure we compare our idealized network to random networks
where all $V_{kl}$ are chosen randomly between 0 and $\sV$ and all
$E_{k}$ are chosen randomly between 0 and $\sG$, hence they have
the same topology. The magnitudes of the errors are similar for the
range considered, but the slopes are different. All the samples show
an error slope of $1$ for $\Delta\tau_{0,\,\mathrm{rel}}$, while
the error slope for $\Delta\tau_{\mathrm{rel}}$ is varying, but in
most parts steeper than the slope of $\Delta\tau_{0,\,\mathrm{rel}}$.
This behavior is closer to the behavior expected from our bounds.
Generally, the agreement is about six orders of magnitude better for
the network $N$ than the network $N_{0}$.

For the ideal highly connected network we derive the quantities used
in Theorem~\ref{thm:relax-bound-1} and \ref{thm:relax-bound-2}
analytically in Appendix~\ref{app:Calculations-for-applications}.
The resulting bounds are
\begin{align*}
\Delta\tau_{\mathrm{rel}} & \leq c_{1}n\sV^{2}\sG^{-2}\\
\Delta\tau_{1,\,\mathrm{rel}} & \leq c_{2}n\sV\sG^{-1}
\end{align*}
for dimension and scaling independent constants $c_{1}$ and $c_{2}$.
The simulation of $M$ has a relatively high error and becomes slow
very fast as $n$ gets larger. Hence, we can only get meaningful results
for $\Delta\tau_{1,\,\mathrm{rel}}$, the relaxation time difference
of networks $N$ and $N_{0}$. The result in Figure~\ref{fig:high-dimensions}
actually shows that the difference increases with slope 2 or proportional
to $n^{2}$. The reason is that in Theorem~\ref{thm:relax-bound-2}
we have the condition$\left\Vert \nu\right\Vert \leq\frac{1}{2}\left\Vert b^{-1}\right\Vert ^{-1}$where
the LHS is proportional to $n$ and the RHS is constant (also discussed
in the Appendix). If we increase the dimension at constant scaling,
this condition and our bound break down. To still get a bound for
large $n$ we would need to readjust the scaling.

\begin{figure}
\begin{centering}
\begingroup
  \makeatletter
  \providecommand\color[2][]{    \GenericError{(gnuplot) \space\space\space\@spaces}{      Package color not loaded in conjunction with
      terminal option `colourtext'    }{See the gnuplot documentation for explanation.    }{Either use 'blacktext' in gnuplot or load the package
      color.sty in LaTeX.}    \renewcommand\color[2][]{}  }  \providecommand\includegraphics[2][]{    \GenericError{(gnuplot) \space\space\space\@spaces}{      Package graphicx or graphics not loaded    }{See the gnuplot documentation for explanation.    }{The gnuplot epslatex terminal needs graphicx.sty or graphics.sty.}    \renewcommand\includegraphics[2][]{}  }  \providecommand\rotatebox[2]{#2}  \@ifundefined{ifGPcolor}{    \newif\ifGPcolor
    \GPcolortrue
  }{}  \@ifundefined{ifGPblacktext}{    \newif\ifGPblacktext
    \GPblacktexttrue
  }{}    \let\gplgaddtomacro\g@addto@macro
    \gdef\gplbacktext{}  \gdef\gplfronttext{}  \makeatother
  \ifGPblacktext
        \def\colorrgb#1{}    \def\colorgray#1{}  \else
        \ifGPcolor
      \def\colorrgb#1{\color[rgb]{#1}}      \def\colorgray#1{\color[gray]{#1}}      \expandafter\def\csname LTw\endcsname{\color{white}}      \expandafter\def\csname LTb\endcsname{\color{black}}      \expandafter\def\csname LTa\endcsname{\color{black}}      \expandafter\def\csname LT0\endcsname{\color[rgb]{1,0,0}}      \expandafter\def\csname LT1\endcsname{\color[rgb]{0,1,0}}      \expandafter\def\csname LT2\endcsname{\color[rgb]{0,0,1}}      \expandafter\def\csname LT3\endcsname{\color[rgb]{1,0,1}}      \expandafter\def\csname LT4\endcsname{\color[rgb]{0,1,1}}      \expandafter\def\csname LT5\endcsname{\color[rgb]{1,1,0}}      \expandafter\def\csname LT6\endcsname{\color[rgb]{0,0,0}}      \expandafter\def\csname LT7\endcsname{\color[rgb]{1,0.3,0}}      \expandafter\def\csname LT8\endcsname{\color[rgb]{0.5,0.5,0.5}}    \else
            \def\colorrgb#1{\color{black}}      \def\colorgray#1{\color[gray]{#1}}      \expandafter\def\csname LTw\endcsname{\color{white}}      \expandafter\def\csname LTb\endcsname{\color{black}}      \expandafter\def\csname LTa\endcsname{\color{black}}      \expandafter\def\csname LT0\endcsname{\color{black}}      \expandafter\def\csname LT1\endcsname{\color{black}}      \expandafter\def\csname LT2\endcsname{\color{black}}      \expandafter\def\csname LT3\endcsname{\color{black}}      \expandafter\def\csname LT4\endcsname{\color{black}}      \expandafter\def\csname LT5\endcsname{\color{black}}      \expandafter\def\csname LT6\endcsname{\color{black}}      \expandafter\def\csname LT7\endcsname{\color{black}}      \expandafter\def\csname LT8\endcsname{\color{black}}    \fi
  \fi
  \setlength{\unitlength}{0.0500bp}  \begin{picture}(7200.00,5760.00)    \gplgaddtomacro\gplbacktext{    }    \gplgaddtomacro\gplfronttext{    }    \gplgaddtomacro\gplbacktext{      \csname LTb\endcsname      \put(588,785){\makebox(0,0)[r]{\strut{}-6}}      \put(588,1204){\makebox(0,0)[r]{\strut{}-4}}      \put(588,1623){\makebox(0,0)[r]{\strut{}-2}}      \put(588,2042){\makebox(0,0)[r]{\strut{}0}}      \put(588,2460){\makebox(0,0)[r]{\strut{}2}}      \put(588,2879){\makebox(0,0)[r]{\strut{}4}}      \colorrgb{0.00,0.00,0.00}      \put(720,356){\makebox(0,0){\strut{}-5}}      \colorrgb{0.00,0.00,0.00}      \put(1543,356){\makebox(0,0){\strut{}-4}}      \colorrgb{0.00,0.00,0.00}      \put(2366,356){\makebox(0,0){\strut{}-3}}      \colorrgb{0.00,0.00,0.00}      \put(3189,356){\makebox(0,0){\strut{}-2}}      \colorrgb{0.00,0.00,0.00}      \put(4011,356){\makebox(0,0){\strut{}-1}}      \colorrgb{0.00,0.00,0.00}      \put(4834,356){\makebox(0,0){\strut{}0}}      \colorrgb{0.00,0.00,0.00}      \put(5657,356){\makebox(0,0){\strut{}1}}      \colorrgb{0.00,0.00,0.00}      \put(6480,356){\makebox(0,0){\strut{}2}}      \colorrgb{0.00,0.00,0.00}      \put(82,1727){\rotatebox{90}{\makebox(0,0){\strut{}$\log_{10}\Delta\tau_\mathrm{rel}$}}}      \colorrgb{0.00,0.00,0.00}      \put(3600,26){\makebox(0,0){\strut{}$\log_{10}V\Gamma^{-1}$}}    }    \gplgaddtomacro\gplfronttext{      \csname LTb\endcsname      \put(2040,2706){\makebox(0,0)[r]{\strut{}$M$-$N$}}      \csname LTb\endcsname      \put(2040,2486){\makebox(0,0)[r]{\strut{}$M$-$N_0$}}    }    \gplgaddtomacro\gplbacktext{      \csname LTb\endcsname      \put(588,3089){\makebox(0,0)[r]{\strut{}-6}}      \put(588,3508){\makebox(0,0)[r]{\strut{}-4}}      \put(588,3927){\makebox(0,0)[r]{\strut{}-2}}      \put(588,4346){\makebox(0,0)[r]{\strut{}0}}      \put(588,4764){\makebox(0,0)[r]{\strut{}2}}      \put(588,5183){\makebox(0,0)[r]{\strut{}4}}      \colorrgb{0.00,0.00,0.00}      \put(720,2660){\makebox(0,0){\strut{}}}      \colorrgb{0.00,0.00,0.00}      \put(1543,2660){\makebox(0,0){\strut{}}}      \colorrgb{0.00,0.00,0.00}      \put(2366,2660){\makebox(0,0){\strut{}}}      \colorrgb{0.00,0.00,0.00}      \put(3189,2660){\makebox(0,0){\strut{}}}      \colorrgb{0.00,0.00,0.00}      \put(4011,2660){\makebox(0,0){\strut{}}}      \colorrgb{0.00,0.00,0.00}      \put(4834,2660){\makebox(0,0){\strut{}}}      \colorrgb{0.00,0.00,0.00}      \put(5657,2660){\makebox(0,0){\strut{}}}      \colorrgb{0.00,0.00,0.00}      \put(6480,2660){\makebox(0,0){\strut{}}}      \colorrgb{0.00,0.00,0.00}      \put(82,4031){\rotatebox{90}{\makebox(0,0){\strut{}$\log_{10}\Delta\tau_\mathrm{rel}$}}}    }    \gplgaddtomacro\gplfronttext{    }    \gplbacktext
    \put(0,0){\includegraphics{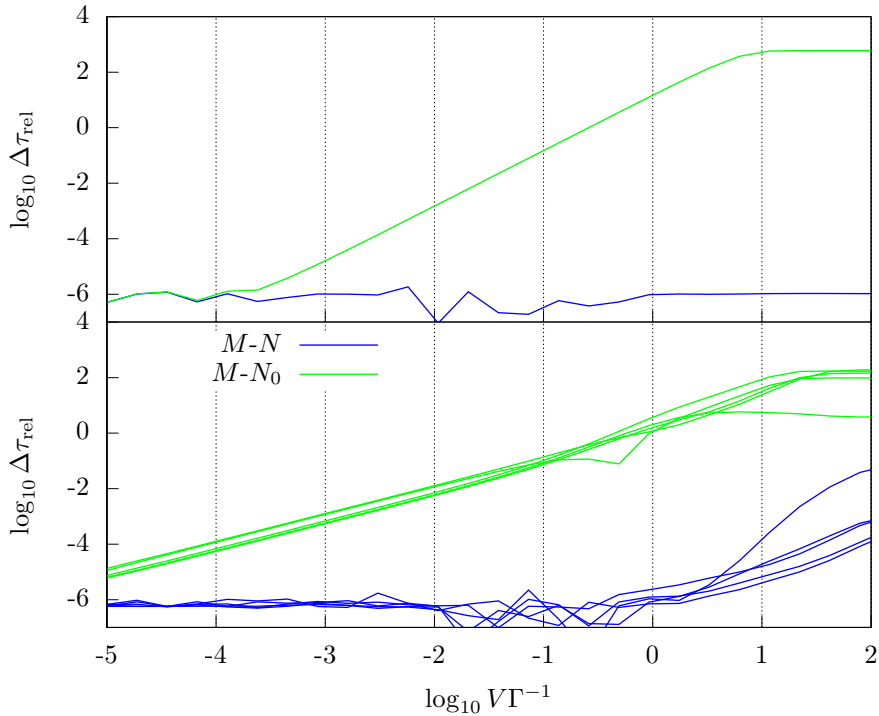}}    \gplfronttext
  \end{picture}\endgroup
\par\end{centering}

\caption{Relative error for the highly connected network\label{fig:highly-connected}}

\end{figure}

\begin{figure}
\begin{centering}
\begingroup
  \makeatletter
  \providecommand\color[2][]{    \GenericError{(gnuplot) \space\space\space\@spaces}{      Package color not loaded in conjunction with
      terminal option `colourtext'    }{See the gnuplot documentation for explanation.    }{Either use 'blacktext' in gnuplot or load the package
      color.sty in LaTeX.}    \renewcommand\color[2][]{}  }  \providecommand\includegraphics[2][]{    \GenericError{(gnuplot) \space\space\space\@spaces}{      Package graphicx or graphics not loaded    }{See the gnuplot documentation for explanation.    }{The gnuplot epslatex terminal needs graphicx.sty or graphics.sty.}    \renewcommand\includegraphics[2][]{}  }  \providecommand\rotatebox[2]{#2}  \@ifundefined{ifGPcolor}{    \newif\ifGPcolor
    \GPcolortrue
  }{}  \@ifundefined{ifGPblacktext}{    \newif\ifGPblacktext
    \GPblacktexttrue
  }{}    \let\gplgaddtomacro\g@addto@macro
    \gdef\gplbacktext{}  \gdef\gplfronttext{}  \makeatother
  \ifGPblacktext
        \def\colorrgb#1{}    \def\colorgray#1{}  \else
        \ifGPcolor
      \def\colorrgb#1{\color[rgb]{#1}}      \def\colorgray#1{\color[gray]{#1}}      \expandafter\def\csname LTw\endcsname{\color{white}}      \expandafter\def\csname LTb\endcsname{\color{black}}      \expandafter\def\csname LTa\endcsname{\color{black}}      \expandafter\def\csname LT0\endcsname{\color[rgb]{1,0,0}}      \expandafter\def\csname LT1\endcsname{\color[rgb]{0,1,0}}      \expandafter\def\csname LT2\endcsname{\color[rgb]{0,0,1}}      \expandafter\def\csname LT3\endcsname{\color[rgb]{1,0,1}}      \expandafter\def\csname LT4\endcsname{\color[rgb]{0,1,1}}      \expandafter\def\csname LT5\endcsname{\color[rgb]{1,1,0}}      \expandafter\def\csname LT6\endcsname{\color[rgb]{0,0,0}}      \expandafter\def\csname LT7\endcsname{\color[rgb]{1,0.3,0}}      \expandafter\def\csname LT8\endcsname{\color[rgb]{0.5,0.5,0.5}}    \else
            \def\colorrgb#1{\color{black}}      \def\colorgray#1{\color[gray]{#1}}      \expandafter\def\csname LTw\endcsname{\color{white}}      \expandafter\def\csname LTb\endcsname{\color{black}}      \expandafter\def\csname LTa\endcsname{\color{black}}      \expandafter\def\csname LT0\endcsname{\color{black}}      \expandafter\def\csname LT1\endcsname{\color{black}}      \expandafter\def\csname LT2\endcsname{\color{black}}      \expandafter\def\csname LT3\endcsname{\color{black}}      \expandafter\def\csname LT4\endcsname{\color{black}}      \expandafter\def\csname LT5\endcsname{\color{black}}      \expandafter\def\csname LT6\endcsname{\color{black}}      \expandafter\def\csname LT7\endcsname{\color{black}}      \expandafter\def\csname LT8\endcsname{\color{black}}    \fi
  \fi
  \setlength{\unitlength}{0.0500bp}  \begin{picture}(7200.00,5760.00)    \gplgaddtomacro\gplbacktext{      \colorrgb{0.00,0.00,0.00}      \put(803,634){\makebox(0,0)[r]{\strut{}-4.5}}      \colorrgb{0.00,0.00,0.00}      \put(803,1304){\makebox(0,0)[r]{\strut{}-4}}      \colorrgb{0.00,0.00,0.00}      \put(803,1975){\makebox(0,0)[r]{\strut{}-3.5}}      \colorrgb{0.00,0.00,0.00}      \put(803,2645){\makebox(0,0)[r]{\strut{}-3}}      \colorrgb{0.00,0.00,0.00}      \put(803,3316){\makebox(0,0)[r]{\strut{}-2.5}}      \colorrgb{0.00,0.00,0.00}      \put(803,3986){\makebox(0,0)[r]{\strut{}-2}}      \colorrgb{0.00,0.00,0.00}      \put(803,4657){\makebox(0,0)[r]{\strut{}-1.5}}      \colorrgb{0.00,0.00,0.00}      \put(803,5327){\makebox(0,0)[r]{\strut{}-1}}      \colorrgb{0.00,0.00,0.00}      \put(935,414){\makebox(0,0){\strut{}0.4}}      \colorrgb{0.00,0.00,0.00}      \put(1865,414){\makebox(0,0){\strut{}0.6}}      \colorrgb{0.00,0.00,0.00}      \put(2795,414){\makebox(0,0){\strut{}0.8}}      \colorrgb{0.00,0.00,0.00}      \put(3725,414){\makebox(0,0){\strut{}1}}      \colorrgb{0.00,0.00,0.00}      \put(4655,414){\makebox(0,0){\strut{}1.2}}      \colorrgb{0.00,0.00,0.00}      \put(5585,414){\makebox(0,0){\strut{}1.4}}      \colorrgb{0.00,0.00,0.00}      \put(6515,414){\makebox(0,0){\strut{}1.6}}      \colorrgb{0.00,0.00,0.00}      \put(33,2980){\rotatebox{90}{\makebox(0,0){\strut{}$\log_{10}\Delta\tau_{1,\,\mathrm{rel}}$}}}      \colorrgb{0.00,0.00,0.00}      \put(3725,84){\makebox(0,0){\strut{}$\log_{10}n$}}    }    \gplgaddtomacro\gplfronttext{      \csname LTb\endcsname      \put(1859,5154){\makebox(0,0)[r]{\strut{}highly}}      \csname LTb\endcsname      \put(1859,4934){\makebox(0,0)[r]{\strut{}chain}}    }    \gplbacktext
    \put(0,0){\includegraphics{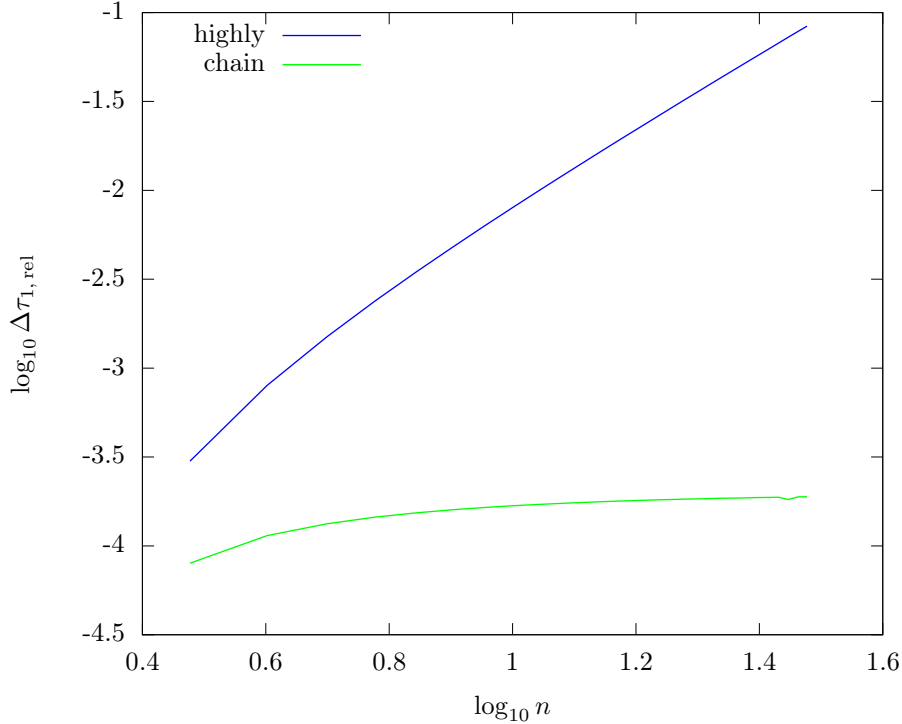}}    \gplfronttext
  \end{picture}\endgroup
\par\end{centering}

\caption{Relative errors between $N$ and $N_{0}$ for the highly connected
network and the cyclical chain with increasing dimension and $\sV=0.01$
and $\sG=1$.\label{fig:high-dimensions}}

\end{figure}

\subsection{Linear network}

Assume the sites are positioned on a circle and only nearest neighbors
interact with strength $\sV$
\[
V_{kl}=\begin{cases}
\sV & |k-l|=1\\
0 & \mathrm{else}
\end{cases}
\]
where we use the equivalence $n\equiv0$. Further $\gamma_{k}=\sG$
and $E_{k}$ such that $E_{kl}=\sG E$ when $|k-l|=1$ which is possible
for $n$ even.

In Figure~\ref{fig:chain} we made a plot of the computed relative
relaxation time differences $\Delta\tau_{\mathrm{rel}}$ and $\Delta\tau_{0,\,\mathrm{rel}}$
for different $\sV\sG^{-1}$ with the initial state localized at site
1. Interestingly the quality of approximation by $N_{0}$ is improved
over the highly connected model, while the quality of approximation
by $N$ has decreased. Also, both models show the same slope of about
$2$. We compare the ideal chain to random chains for which the $V_{kl}$
that equal $\sV$ in the idealized case are instead chosen randomly
between 0 and $\sV$, and all $E_{k}$ are chosen randomly between
0 and $\sG$. We get essentially the same behavior with all slopes
being 2. That hints at a possible improvement of our bound in Theorem~\ref{thm:relax-bound-2}
in the case where the network is a chain, improving the proportionality
from $\sV\sG^{-1}$ to $\sV^{2}\sG^{-2}$. Generally, the agreement
is about five orders of magnitude better for the network $N$ than
the network $N_{0}$.

As in the last section, we can derive the necessary quantities for
our bounds and get

\begin{align*}
\Delta\tau_{\mathrm{rel}} & \leq c_{3}\sV^{2}\sG^{-2}n^{2}\\
\Delta\tau_{1,\,\mathrm{rel}} & \leq c_{4}\sV\sG^{-1}n^{2}
\end{align*}
for dimension and scaling independent constants $c_{3}$ and $c_{4}$.
This time the condition$\left\Vert \nu\right\Vert \leq\frac{1}{2}\left\Vert b^{-1}\right\Vert ^{-1}$
does not break down and the bounds hold for large dimensions as well.
The $n^{2}$ terms are due to the lowest eigenvalue of $N_{0}$ being
proportional to $n^{-2}$. This is a weakness of our strategy to use
the operator norm for our bounds. Better bounds should be possible
when only considering localized exciton as initial state. This initial
state would the a superposition of all the eigenstates on $N_{0}$,
and the average relaxation time would enter the bounds, instead of
the longest relaxation time (the smallest eigenvalue of $N_{0}$).

As above we skip the simulation of $M$ because the error is too large,
and consider $\Delta\tau_{1,\,\mathrm{rel}}$ only. The result in
Figure~\ref{fig:high-dimensions} shows that the difference seems
to approximate a constant value for larger dimensions. So, both our
bounds could be improved for large dimensions.

\begin{figure}
\begin{centering}
\begingroup
  \makeatletter
  \providecommand\color[2][]{    \GenericError{(gnuplot) \space\space\space\@spaces}{      Package color not loaded in conjunction with
      terminal option `colourtext'    }{See the gnuplot documentation for explanation.    }{Either use 'blacktext' in gnuplot or load the package
      color.sty in LaTeX.}    \renewcommand\color[2][]{}  }  \providecommand\includegraphics[2][]{    \GenericError{(gnuplot) \space\space\space\@spaces}{      Package graphicx or graphics not loaded    }{See the gnuplot documentation for explanation.    }{The gnuplot epslatex terminal needs graphicx.sty or graphics.sty.}    \renewcommand\includegraphics[2][]{}  }  \providecommand\rotatebox[2]{#2}  \@ifundefined{ifGPcolor}{    \newif\ifGPcolor
    \GPcolortrue
  }{}  \@ifundefined{ifGPblacktext}{    \newif\ifGPblacktext
    \GPblacktexttrue
  }{}    \let\gplgaddtomacro\g@addto@macro
    \gdef\gplbacktext{}  \gdef\gplfronttext{}  \makeatother
  \ifGPblacktext
        \def\colorrgb#1{}    \def\colorgray#1{}  \else
        \ifGPcolor
      \def\colorrgb#1{\color[rgb]{#1}}      \def\colorgray#1{\color[gray]{#1}}      \expandafter\def\csname LTw\endcsname{\color{white}}      \expandafter\def\csname LTb\endcsname{\color{black}}      \expandafter\def\csname LTa\endcsname{\color{black}}      \expandafter\def\csname LT0\endcsname{\color[rgb]{1,0,0}}      \expandafter\def\csname LT1\endcsname{\color[rgb]{0,1,0}}      \expandafter\def\csname LT2\endcsname{\color[rgb]{0,0,1}}      \expandafter\def\csname LT3\endcsname{\color[rgb]{1,0,1}}      \expandafter\def\csname LT4\endcsname{\color[rgb]{0,1,1}}      \expandafter\def\csname LT5\endcsname{\color[rgb]{1,1,0}}      \expandafter\def\csname LT6\endcsname{\color[rgb]{0,0,0}}      \expandafter\def\csname LT7\endcsname{\color[rgb]{1,0.3,0}}      \expandafter\def\csname LT8\endcsname{\color[rgb]{0.5,0.5,0.5}}    \else
            \def\colorrgb#1{\color{black}}      \def\colorgray#1{\color[gray]{#1}}      \expandafter\def\csname LTw\endcsname{\color{white}}      \expandafter\def\csname LTb\endcsname{\color{black}}      \expandafter\def\csname LTa\endcsname{\color{black}}      \expandafter\def\csname LT0\endcsname{\color{black}}      \expandafter\def\csname LT1\endcsname{\color{black}}      \expandafter\def\csname LT2\endcsname{\color{black}}      \expandafter\def\csname LT3\endcsname{\color{black}}      \expandafter\def\csname LT4\endcsname{\color{black}}      \expandafter\def\csname LT5\endcsname{\color{black}}      \expandafter\def\csname LT6\endcsname{\color{black}}      \expandafter\def\csname LT7\endcsname{\color{black}}      \expandafter\def\csname LT8\endcsname{\color{black}}    \fi
  \fi
  \setlength{\unitlength}{0.0500bp}  \begin{picture}(7200.00,5760.00)    \gplgaddtomacro\gplbacktext{    }    \gplgaddtomacro\gplfronttext{    }    \gplgaddtomacro\gplbacktext{      \csname LTb\endcsname      \put(588,832){\makebox(0,0)[r]{\strut{}-6}}      \put(588,1344){\makebox(0,0)[r]{\strut{}-4}}      \put(588,1855){\makebox(0,0)[r]{\strut{}-2}}      \put(588,2367){\makebox(0,0)[r]{\strut{}0}}      \put(588,2879){\makebox(0,0)[r]{\strut{}2}}      \colorrgb{0.00,0.00,0.00}      \put(720,356){\makebox(0,0){\strut{}-4}}      \colorrgb{0.00,0.00,0.00}      \put(1680,356){\makebox(0,0){\strut{}-3}}      \colorrgb{0.00,0.00,0.00}      \put(2640,356){\makebox(0,0){\strut{}-2}}      \colorrgb{0.00,0.00,0.00}      \put(3600,356){\makebox(0,0){\strut{}-1}}      \colorrgb{0.00,0.00,0.00}      \put(4560,356){\makebox(0,0){\strut{}0}}      \colorrgb{0.00,0.00,0.00}      \put(5520,356){\makebox(0,0){\strut{}1}}      \colorrgb{0.00,0.00,0.00}      \put(6480,356){\makebox(0,0){\strut{}2}}      \colorrgb{0.00,0.00,0.00}      \put(82,1727){\rotatebox{90}{\makebox(0,0){\strut{}$\log_{10}\Delta\tau_\mathrm{rel}$}}}      \colorrgb{0.00,0.00,0.00}      \put(3600,26){\makebox(0,0){\strut{}$\log_{10}V\Gamma^{-1}$}}    }    \gplgaddtomacro\gplfronttext{      \csname LTb\endcsname      \put(2040,2706){\makebox(0,0)[r]{\strut{}$M$-$N$}}      \csname LTb\endcsname      \put(2040,2486){\makebox(0,0)[r]{\strut{}$M$-$N_0$}}    }    \gplgaddtomacro\gplbacktext{      \csname LTb\endcsname      \put(588,3136){\makebox(0,0)[r]{\strut{}-6}}      \put(588,3648){\makebox(0,0)[r]{\strut{}-4}}      \put(588,4159){\makebox(0,0)[r]{\strut{}-2}}      \put(588,4671){\makebox(0,0)[r]{\strut{}0}}      \put(588,5183){\makebox(0,0)[r]{\strut{}2}}      \colorrgb{0.00,0.00,0.00}      \put(720,2660){\makebox(0,0){\strut{}}}      \colorrgb{0.00,0.00,0.00}      \put(1680,2660){\makebox(0,0){\strut{}}}      \colorrgb{0.00,0.00,0.00}      \put(2640,2660){\makebox(0,0){\strut{}}}      \colorrgb{0.00,0.00,0.00}      \put(3600,2660){\makebox(0,0){\strut{}}}      \colorrgb{0.00,0.00,0.00}      \put(4560,2660){\makebox(0,0){\strut{}}}      \colorrgb{0.00,0.00,0.00}      \put(5520,2660){\makebox(0,0){\strut{}}}      \colorrgb{0.00,0.00,0.00}      \put(6480,2660){\makebox(0,0){\strut{}}}      \colorrgb{0.00,0.00,0.00}      \put(82,4031){\rotatebox{90}{\makebox(0,0){\strut{}$\log_{10}\Delta\tau_\mathrm{rel}$}}}    }    \gplgaddtomacro\gplfronttext{    }    \gplbacktext
    \put(0,0){\includegraphics{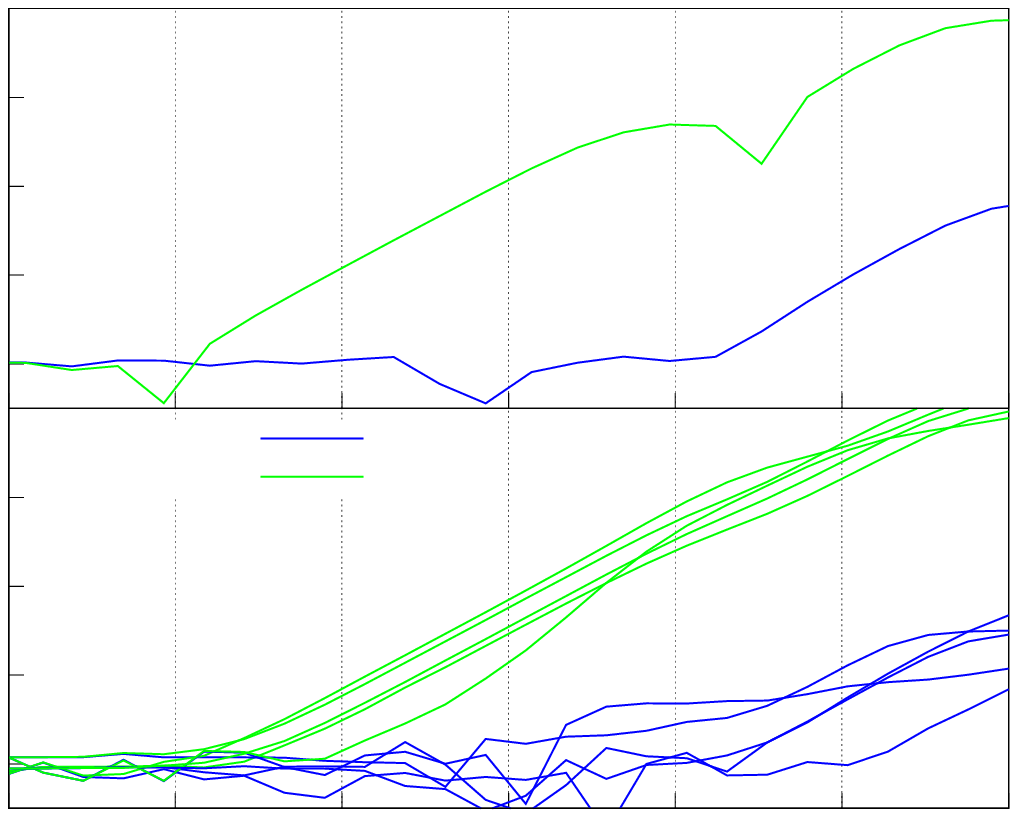}}    \gplfronttext
  \end{picture}\endgroup
\par\end{centering}

\caption{Relative error for the circular chain\label{fig:chain}}
\end{figure}

\subsection{The FMO-complex\label{sub:The-FMO-complex}}

The FMO-complex is pigment-protein with trimer structure. Each monomer
contains seven bacteriochlorophyll $a$ pigments that capture and
transport light. The excitons start out at site 1 or 6 and the trapping
occurs at site 3 \cite{adolphs_how_2006}, we set the initial state
to be
\[
\vec{p}_{0}=(1/2,0,0,0,0,1/2,0)^{\dagger}\,.
\]
We use the same numerical values as \cite{rebentrost_environment-assisted_2009},
with interactions and energies from \cite{cho_exciton_2005}. The
system Hamiltonian is
\[
H+V=\begin{pmatrix}280 & -106 & 8 & -5 & 6 & -8 & -4\\
-106 & 420 & 28 & 6 & 2 & 13 & 1\\
8 & 28 & 0 & -62 & -1 & -9 & 17\\
-5 & 6 & -62 & 175 & -70 & -19 & -57\\
6 & 2 & -1 & -70 & 320 & 40 & -2\\
-8 & 13 & -9 & -19 & 40 & 360 & 32\\
-4 & 1 & 17 & -57 & -2 & 32 & 260
\end{pmatrix}
\]
with all the numbers in $cm^{-1}$ (or $2.9978\cdot10^{10}s^{-1}$).
Exciton recombination at rate $\kappa=1ns^{-1}$ and reaction center
trapping at rate $\kappa_{3}=1ps^{-1}$ enter the anti-hermitian operator
\[
A=-\frac{i}{2}\left(\sum_{k}\kappa\ket k\bra k+\kappa_{3}\ket 3\bra 3\right)\,.
\]
We use the same dephasing rate for every site $\gamma_{k}=\gamma$,
and vary $\gamma$ from $10^{-3}$ to $10^{5}cm^{-1}$. Efficiency
is calculated as
\[
f=\kappa_{3}\int_{0}^{\infty}\rho_{33}(t)\dt
\]
we calculated $f$ for the three models in Figure~\ref{fig:FMO-efficiency}.
Peak efficiency is reached for $\gamma\approx170cm^{-1}$ close to
the average energy gap along the chain which is $146cm^{-1}$. The
approximation $N$ has less than $1\%$ error, even for the lowest
$\gamma$ used, and the approximation $N_{0}$ gets below $1\%$ error
for $\gamma\approx2cm^{-1}$. Comparing this to our bounds we have
\[
\left\Vert a\right\Vert =\left\Vert a^{\dagger}\right\Vert =215cm^{-1}
\]
and for large $\gamma$ 
\[
\left\Vert b^{-1}\right\Vert ^{-1}=\gamma\,.
\]
The numerical factor $\beta$ is changing because of the changing
ratio between energies and dephasing, for large $\gamma$ however
it is approximately equal to 100. Hence, our bound becomes
\[
\Delta\tau_{\mathrm{rel}}\lesssim100\left(215cm^{-1}\gamma^{-1}\right)^{2}\,.
\]
The $1\%$ error margin is reached only when $\gamma=21500cm^{-1}$,
so our numerical factors could certainly be much improved. But this
is not unexpected, since our main goal was to find the leading behavior
in $\sV\sG^{-1}$.

We give $N_{0}$ for maximal transfer efficiency
\[
N_{0}(\gamma=170cm^{-1})=\begin{pmatrix}-80 & 79 & 0 & 0 & 0 & 1 & 0\\
79 & -82 & 1 & 0 & 0 & 2 & 0\\
0 & 1 & -58 & 22 & 0 & 0 & 1\\
0 & 0 & 22 & -88 & 33 & 2 & 31\\
0 & 0 & 0 & 33 & -52 & 18 & 0\\
1 & 2 & 0 & 2 & 18 & -31 & 9\\
0 & 0 & 1 & 31 & 0 & 9 & -41
\end{pmatrix}\,.
\]
It is interesting that the rate between sites 2 and 3 is actually
smaller than the rate between sites 2 and 6 even though $\left|V_{23}\right|>\left|V_{26}\right|$.
The reason is the large energy gap between sites 2 and 3 of $420cm^{-1}$
while sites 2 and 6 have an energy gap of $60cm^{-1}$. However, the
values for site energies are still up to some debate \cite{adolphs_how_2006,cho_exciton_2005},
and small changes can easily turn this behavior to the opposite again.

\section{Resolvent difference bounds\label{sec:Resolvent-difference-bounds}}

The following three Lemmas are the main technical parts of our bounds.
They all consist of bounding the operator norm of the resolvent difference
\[
S(z)=\frac{1}{z-a^{\dagger}(b-z)^{-1}a}-\frac{1}{z-a^{\dagger}b^{-1}a}
\]
for different values of $z$. Conceptually the bounding procedure
is simple, we only employ the inverse bounds introduced in \ref{sub:Inverse-bounds}.
Loosely speaking, if $\left|z\right|<\sG$ we can expand $(b-z)^{-1}$
and then the two terms in $S(z)$ only have a small difference in
the denominator, so, using another inverse bound, they almost cancel.
If $\left|z\right|>\sG$ then $\left|z\right|\gg\left\Vert a^{\dagger}b^{-1}a\right\Vert $
and we can directly use the second step from the case $\left|z\right|<\Gamma$.

Of course we also have to keep in mind where the poles of $S(z)$
are. According to Proposition~\ref{prop:spectral-properties} $(z-N)^{-1}$
has poles on the real axis below $-\mu$ which move according to the
scaling $\sV^{2}\sG^{-1}$. On the other hand $(z-a^{\dagger}(b-z)^{-1}a)^{-1}$
has poles close to the poles of $(z-N)^{-1}$ that approximately cancel
each other, but it also has poles close to the eigenvalues of $b$
which are approximately $\alpha_{ij}=-\gamma_{ij}+iE_{ij}$ and $\bar{\alpha}_{ij}$,
scaling like $\sG$. Comparing the two sets of poles, the $b$-poles
are much further to the left (negative real values) than the $N$-poles
because $\sG\gg\sV^{2}\sG^{-1}$. Our lemma steer clear of this poles
by keeping $\re z\geq-\mu/2$.

Lemma~\ref{lem:decay-bounds-right-plane} contains bounds for $\re z\geq0$
which on the one hand ensures there are no poles on the right side
of the complex plane, and on the other hand we use the bounds for
$z=iy$ to bound the relaxation time. Lemma~\ref{lem:decay-bounds-sliver}
contains bounds for the region $-\mu/2\leq\re z\leq0$ the bounds
are derived in a similar fashion as in Lemma~\ref{lem:decay-bounds-right-plane},
but there are some additional complications.

\subsection{Bounds in the right half plane}
\begin{lem}
\label{lem:decay-bounds-right-plane}If $\sV\sG^{-1}$ is small enough
and $\re z\geq0$ then $S(z)$ is bounded by
\begin{enumerate}
\item $\left\Vert S(z)\right\Vert \leq4\kappa\mu^{-2}|z|$ if $|z|\leq\alpha$,
where $\alpha\propto\sG$ depends on $a$ and $b$,
\item $\left\Vert S(z)\right\Vert \leq4\beta\kappa|z|^{-1}$ for any $z$
with $\re z\geq0$, where $\beta$ is a scaling independent constant
depending on $a$ and $b$.
\end{enumerate}
\end{lem}
\begin{proof}
1. Assume $\re z\geq0$ and $|z|\leq\alpha\propto\sG$, where
\begin{equation}
\alpha=\min\left\{ \frac{1}{2}\left\Vert b^{-1}\right\Vert ^{-1},\,\frac{1}{4}\kappa^{-1}\mu\right\} \,.\label{eq:alpha-def}
\end{equation}
Set
\[
X=(b-z)^{-1}-b^{-1}
\]
and because $|z|\leq\frac{1}{2}\left\Vert b^{-1}\right\Vert ^{-1}$
we can use (\ref{eq:inverse-bound}) and have
\[
\left\Vert X\right\Vert \leq2\left\Vert b^{-1}\right\Vert ^{2}|z|\,.
\]
Rewrite
\[
S(z)=(z-a^{\dagger}b^{-1}a-a^{\dagger}Xa)^{-1}-(z-a^{\dagger}b^{-1}a)^{-1}\,.
\]
To use (\ref{eq:inverse-bound}) on this expression notice that
\[
|z|\leq\frac{1}{4}\kappa^{-1}\mu
\]
and therefore
\begin{align*}
\left\Vert a^{\dagger}Xa\right\Vert  & \leq2\kappa|z|\\
 & \leq\frac{1}{2}\mu\\
 & \leq\frac{1}{2}\left\Vert (z-a^{\dagger}b^{-1}a)^{-1}\right\Vert ^{-1}
\end{align*}
where (\ref{eq:inv-bound-const}) was applied in the last step, using
the fact that $a^{\dagger}b^{-1}a$ is self-adjoint from Proposition~\ref{prop:spectral-properties}.
This is just the condition for the bound
\begin{align*}
\left\Vert S(z)\right\Vert  & \leq2\left\Vert (z-a^{\dagger}b^{-1}a)^{-1}\right\Vert ^{2}\no{a^{\dagger}Xa}\\
 & \leq4\kappa\left\Vert (z-a^{\dagger}b^{-1}a)^{-1}\right\Vert ^{2}|z|
\end{align*}
again using (\ref{eq:inv-bound-const}) and also (\ref{eq:inv-bound-yinv})
we get the bounds
\begin{align}
\left\Vert S(z)\right\Vert  & \leq4\kappa\mu^{-2}|z|\nonumber \\
\left\Vert S(z)\right\Vert  & \leq4\kappa|z|^{-1}\label{eq:y-1-bound-1}
\end{align}
for $|z|\leq\alpha$. The first bound is bound 1 of the Lemma, the
second bound will be used below.\medskip{}

2. We now derive a bound when $|z|\geq\alpha$ and $\re z\geq0$,
we will combine it with (\ref{eq:y-1-bound-1}) to receive bound 2
for all $z\in\R$. If $\sV\sG^{-1}$ is small enough then we have
\begin{align*}
\left\Vert a^{\dagger}b^{-1}a\right\Vert  & \leq\frac{\alpha}{2}\leq\frac{1}{2}|z|\\
\left\Vert a^{\dagger}(b-z)^{-1}a\right\Vert  & \leq\frac{\alpha}{2}\leq\frac{1}{2}|z|\,.
\end{align*}

Where the latter inequality uses the fact that the spectrum of $b$
approaches the spectrum of $b_{0}$ as $\sV\sG^{-1}$ becomes small,
and the spectrum of $b_{0}$, which is $-\gamma_{ij}\pm iE_{ij}$,
has negative real part $-\gamma_{ij}<0$. The last two inequalities
are the conditions to use (\ref{eq:inverse-bound}) and get the two
bounds
\begin{align*}
\left\Vert (z-a^{\dagger}(b-z)^{-1}a)^{-1}-z^{-1}\right\Vert  & \leq2|z|^{-2}\left\Vert a^{\dagger}(b-z)^{-1}a\right\Vert \\
\left\Vert (z-a^{\dagger}b^{-1}a)^{-1}-z^{-1}\right\Vert  & \leq2|z|^{-2}\left\Vert a^{\dagger}b^{-1}a\right\Vert 
\end{align*}
set
\begin{equation}
b_{\min}=\min\,\left\{ \left|\re\lambda\right|\middle|\lambda\in\spec b\right\} \propto\sG\label{eq:b-min}
\end{equation}
the closest any eigenvalue of $b$ gets to the imaginary axis. Then
$\left\Vert b^{-1}\right\Vert \leq b_{\min}^{-1}$ and $\left\Vert (b-z)^{-1}\right\Vert \leq b_{\min}^{-1}$
so 
\begin{align*}
\left\Vert S(z)\right\Vert  & \leq4|z|^{-2}\left\Vert a\right\Vert ^{2}d^{-1}\,.
\end{align*}

Comparing to (\ref{eq:y-1-bound-1}) with
\[
\beta=\max\left\{ 1,\,1/\left(\alpha b_{\min}\left\Vert b^{-1}\right\Vert ^{2}\right)\right\} \propto1
\]
we have
\[
4|z|^{-2}\left\Vert a\right\Vert ^{2}b_{\min}^{-1}\leq\beta\cdot4\kappa|z|^{-1}
\]
for $|z|\geq\alpha$ and therefore
\begin{equation}
\left\Vert S(z)\right\Vert \leq4\beta\kappa|z|^{-1}\label{eq:better-y-1-bound-1}
\end{equation}
for all $z$ with $\re z\geq0$. Which is bound 2 of the Lemma.
\end{proof}

\subsection{Bounds parallel to the imaginary axis}

The following Lemma establishes bounds along the imaginary axis $z=iy-\tilde{\mu}$.
These bounds are used to prove the evolution bounds.
\begin{lem}
\label{lem:decay-bounds-sliver}If we choose $\sV\sG^{-1}$ small
enough then for $0\leq\tilde{\mu}\leq\mu/2$ the resolvent difference
$S(iy-\tilde{\mu})$ is bounded by
\begin{enumerate}
\item $\left\Vert S(iy-\tilde{\mu})\right\Vert \leq16\kappa\mu^{-2}|iy-\tilde{\mu}|$
and $\left\Vert S(iy-\tilde{\mu})\right\Vert \leq4\kappa|y|^{-2}|iy-\tilde{\mu}|$
if $|y|\leq\hat{\alpha}$, where $\hat{\alpha}\propto\sG$ depends
on $a$ and $b$,
\item $\left\Vert S(y)\right\Vert \leq4|y|^{-2}\left\Vert a\right\Vert ^{2}(b_{\min}-\tilde{\mu})^{-1}$
for $|y|>\hat{\alpha}$ with $b_{\min}\propto\sG$.
\end{enumerate}
\end{lem}
\begin{proof}
We proceed almost identically as in the proof of Lemma~\ref{lem:decay-bounds-right-plane}
using the inverse bounds \ref{eq:inverse-bound}, (\ref{eq:inv-bound-const})
and (\ref{eq:inv-bound-yinv}) for the same parts of the resolvent
terms. \medskip{}

1. We use the $\alpha$ from (\ref{eq:alpha-def}) to define
\[
\hat{\alpha}=\min\left\{ \frac{1}{2}\left\Vert b^{-1}\right\Vert ^{-1},\,\frac{1}{8}\kappa^{-1}\mu\right\} -\mu
\]
notice that the scaling $\hat{\alpha}\propto\sG$ is only approximate
and that $\sV\sG^{-1}$ needs to be small enough such that $\hat{\alpha}>0$.
Now require $|y|\leq\hat{\alpha}\propto\sG$. Set
\[
X=(b-iy+\tilde{\mu})^{-1}-b^{-1}
\]
and because we have 
\begin{align*}
|iy-\tilde{\mu}| & \leq|y|+\tilde{\mu}\\
 & \leq\frac{1}{2}\left\Vert b^{-1}\right\Vert ^{-1}-\mu+\tilde{\mu}\\
 & \leq\frac{1}{2}\left\Vert b^{-1}\right\Vert ^{-1}
\end{align*}
we can use (\ref{eq:inverse-bound}) to get the bound
\[
\left\Vert X\right\Vert \leq2\left\Vert b^{-1}\right\Vert ^{2}|iy-\tilde{\mu}|\,.
\]
Rewrite
\[
S(iy-\tilde{\mu})=(iy-\tilde{\mu}-a^{\dagger}b^{-1}a-a^{\dagger}Xa)^{-1}-(iy-\tilde{\mu}-a^{\dagger}b^{-1}a)^{-1}\,.
\]
To use (\ref{eq:inverse-bound}) on this expression notice that we
have
\begin{align*}
|iy-\tilde{\mu}| & \leq|y|+\tilde{\mu}\\
 & \leq\left(\frac{1}{8}\kappa^{-1}\mu-\mu\right)+\tilde{\mu}\\
 & \leq\frac{1}{8}\kappa^{-1}\mu
\end{align*}
and therefore
\begin{align*}
\left\Vert a^{\dagger}Xa\right\Vert  & \leq2\kappa|iy-\tilde{\mu}|\\
 & \leq\frac{1}{4}\mu\\
 & \leq\frac{1}{2}\left\Vert (iy-\tilde{\mu}-a^{\dagger}b^{-1}a)^{-1}\right\Vert ^{-1}
\end{align*}
where (\ref{eq:inv-bound-const}) was applied in the last step, using
the fact that $\tilde{\mu}+a^{\dagger}b^{-1}a\leq-\mu/2$ from Proposition~\ref{prop:spectral-properties}.
This is just the condition for the bound
\begin{align*}
\left\Vert S(iy-\tilde{\mu})\right\Vert  & \leq2\left\Vert (iy-\tilde{\mu}-a^{\dagger}b^{-1}a)^{-1}\right\Vert ^{2}\no{a^{\dagger}Xa}\\
 & \leq4\kappa\left\Vert (iy-\tilde{\mu}-a^{\dagger}b^{-1}a)^{-1}\right\Vert ^{2}|iy-\tilde{\mu}|
\end{align*}
again using (\ref{eq:inv-bound-const}) and also (\ref{eq:inv-bound-yinv})
we get the bounds
\begin{align*}
\left\Vert S(iy-\tilde{\mu})\right\Vert  & \leq16\kappa\mu^{-2}|iy-\tilde{\mu}|\\
\left\Vert S(iy-\tilde{\mu})\right\Vert  & \leq4\kappa|y|^{-2}|iy-\tilde{\mu}|
\end{align*}
for $|y|\leq\hat{\alpha}$. These are the bounds in part 1 of our
Lemma.\medskip{}

2. We now derive a bound when $|y|\geq\hat{\alpha}$. If $\sV\sG^{-1}$
is small enough then
\begin{align*}
\left\Vert a^{\dagger}b^{-1}a+\tilde{\mu}\right\Vert  & \leq\frac{\hat{\alpha}}{2}\leq\frac{1}{2}|y|\\
\left\Vert a^{\dagger}(b-iy+\tilde{\mu})^{-1}a+\tilde{\mu}\right\Vert  & \leq\frac{\hat{\alpha}}{2}\leq\frac{1}{2}|y|\,.
\end{align*}
The last two inequalities are the conditions to use (\ref{eq:inverse-bound})
and get the two bounds
\begin{align*}
\left\Vert (iy-\tilde{\mu}-a^{\dagger}(b-iy+\tilde{\mu})^{-1}a)^{-1}-(iy-\tilde{\mu})^{-1}\right\Vert  & \leq2|y|^{-2}\left\Vert a^{\dagger}(b-iy+\tilde{\mu})^{-1}a\right\Vert \\
\left\Vert (iy-\tilde{\mu}-a^{\dagger}b^{-1}a)^{-1}-(iy-\tilde{\mu})^{-1}\right\Vert  & \leq2|y|^{-2}\left\Vert a^{\dagger}b^{-1}a\right\Vert \,.
\end{align*}
Use $b_{\min}$ from (\ref{eq:b-min}), giving
\begin{align*}
\left\Vert (b-iy+\tilde{\mu})^{-1}\right\Vert  & \leq(b_{\min}-\tilde{\mu})^{-1}\\
\left\Vert (b-iy+\tilde{\mu})^{-1}\right\Vert  & \leq(b_{\min}-\tilde{\mu})^{-1}
\end{align*}
and so
\begin{align*}
\left\Vert S(iy-\tilde{\mu})\right\Vert  & \leq4|y|^{-2}\left\Vert a\right\Vert ^{2}(b_{\min}-\tilde{\mu})^{-1}\,.
\end{align*}
for $|y|>\hat{\alpha}$, which is the bound in part 2 of our Lemma.\end{proof}

\section{Conclusion}

We studied to kinetic networks that approximate the energy transfer
in a quantum network subject to dephasing. The first network $N_{0}$
derives its rates only from nearest neighbor interactions, while the
second $N$ includes higher order corrections. We proved that the
relaxation times are proportional to $\sV\sG^{-1}$ and $\sV^{2}\sG^{-2}$
respectively. Hence, the approximations are good if the interaction
gets weak, or the dephasing and/or energy gaps get large. In the case
of the FMO complex, both kinetic networks are good approximations
in the regime of dephasing-assisted energy transfer. With simulations
we found that the more complex kinetic network $N$ provides approximations
with a percentage error 5-6 magnitudes smaller than the simple kinetic
network.

The study of these approximations could be extended in several ways.
First, one could study the higher order corrections involved in $N$.
Second, when the interactions $V_{kl}$ are complex, $N$ can be non-symmetric,
meaning population exchange between sites is directed, this might
relate to coherent cancellations along loops as mentioned in \cite{cao_optimization_2009}.
And finally, it would be interesting how our method of splitting population
and coherence space to achieve kinetic network approximations could
be generalized to other quantum networks and how it relates to existing
models to approximate coherent evolution with incoherent statistical
evolution.

\section{Acknowledgments}

I want to thank Chris King for his support, ideas and many useful
discussions.

\appendix

\section{Three sites\label{app:Three-sites}}

In the following we write out parts of the master equation (\ref{eq:master-equation-no-anti})
for the case $n=3$ and then derive the form of the matrix $M$. Then
we explain how to generalize that form to higher $n$. For simplicity
of notation we omit the scaling factors $\sG$ and $\sV$, until we
reach a block matrix expression. First note that with a standard calculation
one finds $\L(\rho)$ to decrease the coherences in the manner
\[
\left(\L(\rho)\right)_{kl}=-\gamma_{kl}\rho_{kl}
\]
where $k\neq l$ and $\gamma_{kl}=\frac{1}{2}(\gamma_{k}+\gamma_{l})$
and $\left(\L(\rho)\right)_{kk}=0$. This gives a diagonal contribution
$-\gamma_{kl}$ in the diagonal of the two rows corresponding to the
real and imaginary part of $\rho_{kl}$.

Now, we evaluate the commutator
\[
\left[\begin{pmatrix}E_{1} & V_{12} & V_{13}\\
V_{21} & E_{2} & V_{23}\\
V_{31} & V_{32} & E_{3}
\end{pmatrix},\begin{pmatrix}\rho_{11} & \rho_{12} & \rho_{13}\\
\rho_{21} & \rho_{22} & \rho_{23}\\
\rho_{31} & \rho_{32} & \rho_{33}
\end{pmatrix}\right]\,.
\]
From the 1x1 entry we get
\begin{align*}
\dot{\rho}_{11} & =-i(E_{1}\rho_{11}+V_{12}\rho_{21}+V_{13}\rho_{31}-E_{1}\rho_{11}-V_{21}\rho_{12}-V_{31}\rho_{13})\\
 & =-i(V_{12}\bar{\rho}_{12}+V_{13}\bar{\rho}_{13}-\overline{V}_{12}\rho_{12}-\overline{V}_{13}\rho_{13})\\
 & =2\im(V_{12}\bar{\rho}_{12}+V_{13}\bar{\rho}_{13})\\
 & =2(-V_{12}^{r}\rho_{12}^{i}+V_{12}^{i}\rho_{12}^{r}-V_{13}^{r}\rho_{13}^{i}+V_{13}^{i}\rho_{13}^{r})
\end{align*}
where superscripts $r$ and $i$ are shortcuts for real and imaginary
parts, and from the 1x2 entry we get
\begin{align*}
\dot{\rho}_{12} & =-i(E_{1}\rho_{12}+V_{12}\rho_{22}+V_{13}\rho_{32}-V_{12}\rho_{11}-E_{2}\rho_{12}-V_{32}\rho_{13})-\gamma_{12}\rho_{12}\\
 & =-i((E_{1}-E_{2})\rho_{12}+V_{12}\rho_{22}-V_{12}\rho_{11}+V_{13}\bar{\rho}_{23}-\overline{V}_{23}\rho_{13})-\gamma_{12}\rho_{12}
\end{align*}
with real and imaginary parts
\begin{align*}
\dot{\rho}_{12}^{r} & =-V_{12}^{i}\rho_{11}+V_{12}^{i}\rho_{22}-\gamma_{12}\rho_{12}^{r}+(E_{1}-E_{2})\rho_{12}^{i}+V_{23}^{i}\rho_{13}^{r}-V_{23}^{r}\rho_{13}^{i}+V_{13}^{i}\rho_{23}^{r}-V_{13}^{r}\rho_{23}^{i}\\
\dot{\rho}_{12}^{i} & =V_{12}^{r}\rho_{11}-V_{12}^{r}\rho_{22}-(E_{1}-E_{2})\rho_{12}^{r}-\gamma_{12}\rho_{12}^{i}+V_{23}^{r}\rho_{13}^{r}+V_{23}^{i}\rho_{13}^{i}-V_{13}^{r}\rho_{23}^{r}-V_{13}^{i}\rho_{23}^{i}\,.
\end{align*}
From these results we can read off lines 1, 4 and 5 of the following
matrix and fill in the remaining lines in the same fashion

\[
M=\left(\begin{array}{ccc|cccccc}
 &  &  & \sqrt{2}V_{12}^{i} & -\sqrt{2}V_{12}^{r} & \sqrt{2}V_{13}^{i} & -\sqrt{2}V_{13}^{r}\\
 & 0 &  & -\sqrt{2}V_{12}^{i} & \sqrt{2}V_{12}^{r} &  &  & \sqrt{2}V_{23}^{i} & -\sqrt{2}V_{23}^{r}\\
 &  &  &  &  & -\sqrt{2}V_{13}^{i} & \sqrt{2}V_{13}^{r} & -\sqrt{2}V_{23}^{i} & \sqrt{2}V_{23}^{r}\\
\hline -\sqrt{2}V_{12}^{i} & \sqrt{2}V_{12}^{i} &  & -\gamma_{12} & E_{12} & V_{23}^{i} & -V_{23}^{r} & V_{13}^{i} & -V_{13}^{r}\\
\sqrt{2}V_{12}^{r} & -\sqrt{2}V_{12}^{r} &  & -E_{12} & -\gamma_{12} & V_{23}^{r} & V_{23}^{i} & -V_{13}^{r} & -V_{13}^{i}\\
-\sqrt{2}V_{13}^{i} &  & \sqrt{2}V_{13}^{i} & -V_{23}^{i} & -V_{23}^{r} & -\gamma_{13} & E_{13} & V_{12}^{i} & V_{12}^{r}\\
\sqrt{2}V_{13}^{r} &  & -\sqrt{2}V_{13}^{r} & V_{23}^{r} & -V_{23}^{i} & -E_{13} & -\gamma_{13} & -V_{12}^{r} & V_{12}^{i}\\
 & -\sqrt{2}V_{23}^{i} & \sqrt{2}V_{23}^{i} & -V_{13}^{i} & V_{13}^{r} & -V_{12}^{i} & V_{12}^{r} & -\gamma_{23} & E_{23}\\
 & \sqrt{2}V_{23}^{r} & -\sqrt{2}V_{23}^{r} & V_{13}^{r} & V_{13}^{i} & -V_{12}^{r} & -V_{12}^{i} & -E_{23} & -\gamma_{23}
\end{array}\right)
\]
where we define $E_{ij}=E_{i}-E_{j}$. Written as a block matrix
\[
M=\begin{pmatrix}0 & -a^{\dagger}\\
a & b
\end{pmatrix}
\]
one can see the explicit form of the matrices $a$, and $b$. Remember
that we also separated $b$ into two parts. We set the 2x2-block diagonal
that scales like $\sG$ (the $E_{ij}$ and $\gamma_{ij}$ entries)
to be $b_{0}$ and we set the block-off-diagonal that scales like
$\sV$ (all the $V_{ij}$ entries) to be $\nu$. So $b=b_{0}+\nu$.

In \ref{sub:kinetic-network-N0} in (\ref{eq:U-transformation}) we
defined a transformation $U$ to diagonalize $b_{0}$, if we extend
this transformation to the entire space $P\oplus C$ as
\[
\hat{U}=\one_{n}\oplus U
\]
we can apply it to $M$ directly and get
\[
\tilde{M}=\hat{U}^{\dagger}M\hat{U}=\left(\begin{array}{ccc|cccccc}
 &  &  & -V_{12} & -\overline{V}_{12} & -V_{13} & -\overline{V}_{13}\\
 & 0 &  & V_{12} & \overline{V}_{12} &  &  & -V_{23} & -\overline{V}_{23}\\
 &  &  &  &  & V_{13} & \overline{V}_{13} & V_{23} & \overline{V}_{23}\\
\hline \overline{V}_{12} & -\overline{V}_{12} &  & \alpha_{12} &  & -iV_{23} &  &  & -i\overline{V}_{13}\\
V_{12} & -V_{12} &  &  & \overline{\alpha}_{12} &  & i\overline{V}_{23} & iV_{13}\\
\overline{V}_{13} &  & -\overline{V}_{13} & -i\overline{V}_{23} &  & \alpha_{13} &  & i\overline{V}_{12}\\
V_{13} &  & -V_{13} &  & iV_{23} &  & \overline{\alpha}_{13} &  & -iV_{12}\\
 & \overline{V}_{23} & -\overline{V}_{23} &  & i\overline{V}_{13} & iV_{12} &  & \alpha_{23}\\
 & V_{23} & -V_{23} & -iV_{13} &  &  & -i\overline{V}_{12} &  & \overline{\alpha}_{23}
\end{array}\right)
\]
where $\alpha_{ij}=-\gamma_{ij}+iE_{ij}$. This new matrix consists
of the matrices $\tilde{a}$ and $\tilde{b}_{0}$ also introduced
in \ref{sub:kinetic-network-N0} 
\[
\tilde{M}=\begin{pmatrix}0 & -\tilde{a}^{\dagger}\\
\tilde{a} & \tilde{b}
\end{pmatrix}
\]
where 
\begin{align*}
\tilde{b} & =U^{\dagger}bU\\
 & =\tilde{b}_{0}+\tilde{\nu}
\end{align*}
with $\tilde{\nu}=U^{\dagger}\nu U$. The two kinetic networks are
\begin{align*}
N_{0} & =\tilde{a}^{\dagger}\tilde{b}_{0}^{-1}\tilde{a}\\
N & =\tilde{a}^{\dagger}\tilde{b}^{-1}\tilde{a}
\end{align*}
which also holds with all the tildes removed.

It is straightforward to generalize the matrices $\tilde{a}$ and
$\tilde{b}_{0}$ to $n>3$. Matrix $\tilde{a}$ connects the population
of site $k$ to the coherences between site $k$ and any other site
$l$ with strength $V_{kl}$, and matrix $\tilde{b}_{0}$ is a diagonal
matrix with entries $\alpha_{ij}$ and $\overline{\alpha}_{ij}$.
A bit more complicated is the matrix $\tilde{\nu}$ it is described
in the next subsection.

\section{General construction\label{app:Constructing-nu}}

Here we give a description of how to find $\tilde{a}$, $\tilde{b}_{0}$
and $\tilde{\nu}$ for general $n$. We number the $n$ dimensions
of population space $P$ with $k$ where $k=1,2,\dots n$ and the
$(n^{2}-n)$ dimensions of coherence space $C$ with $kl$ and $\overline{kl}$
where $k<l$ are numbers from 1 to $n$. According to the order defined
in \ref{sub:Converting-the-master} the first few dimensions of $C$
are called $12$, $\overline{12}$, $13$, ..., $23$, $\overline{23}$,
$24$, etc. .

\subsection{Constructing $\tilde{a}$ and $\tilde{b}_{0}$}

Matrix $\tilde{a}$ is an $n\times(n^{2}-n)$ complex matrix, with
the only nonzero entries
\begin{align*}
\tilde{a}_{k,kl}= & \bar{V}_{kl}=-\tilde{a}_{k,lk}\\
\tilde{a}_{k,\overline{kl}}= & V_{kl}=-\tilde{a}_{k,\overline{lk}}\,,
\end{align*}
hence in every column there are only $(n-1)$ nonzero entries.

Matrix $\tilde{b}_{0}$ is diagonal with entries
\begin{align*}
\left(\tilde{b}_{0}\right)_{kl,kl} & =-\gamma_{kl}+iE_{kl}\\
\left(\tilde{b}_{0}\right)_{\overline{kl},\overline{kl}} & =-\gamma_{kl}-iE_{kl}\,.
\end{align*}

\subsection{Constructing $\tilde{\nu}$}

The matrix $\tilde{\nu}=U^{\dagger}\nu U$ for any $n$ is a somewhat
complicated pattern of entries $V_{kl}$, signs and complex conjugates.
It connects coherences between sites $k$ and $l$ with coherences
between sites $k$ and $m$ with the strength $V_{lm}$. Entries of
$\tilde{\nu}$ are only non-zero if one number of the two double indices
match with further conditions on their conjugation. Table~\ref{tab:nu-entries}
shows the rules for the nonzero entries.
\begin{table}
\begin{tabular}{|c|c|c|c|c|c|}
\hline 
row & column & entry & row & column & entry\tabularnewline
\hline 
\hline 
$kl$ & $km$ & $-iV_{lm}$ & $\overline{kl}$ & $\overline{km}$ & $i\bar{V}_{lm}$\tabularnewline
\hline 
$lk$ & $mk$ & $i\bar{V}_{lm}$ & $\overline{lk}$ & $\overline{mk}$ & $-iV_{lm}$\tabularnewline
\hline 
$lk$ & $\overline{km}$ & $-i\bar{V}_{lm}$ & $\overline{lk}$ & $km$ & $iV_{lm}$\tabularnewline
\hline 
\end{tabular}\caption{The non-zero entries of $\tilde{\nu}$, always $l\neq m$\label{tab:nu-entries}}
\end{table}

\section{Calculations for applications\label{app:Calculations-for-applications}}

\subsection{Highly connected network}

Assume all sites are equally interacting, and have the same energies
and dephasing rates
\begin{align*}
V_{kl} & =\sV\\
E_{k} & =0\\
\gamma_{k} & =\sG\,.
\end{align*}
Then every column in $a$ has $2(n-1)$ non-zero entries all equal
to $\sV$. A simple calculation shows that
\[
a^{\dagger}a=2n\sV^{2}\left(\one_{n}-n\vec{e}\vec{e}^{\dagger}\right)
\]
so for any $\vec{v}\in I$ we have $a^{\dagger}a\vec{v}=2\sV^{2}n\vec{v}$,
hence 
\[
\left\Vert a\right\Vert =\sqrt{2n}\sV\,.
\]

Obviously,$b=-\sG\one_{C}$ and $\left\Vert b_{0}^{-1}\right\Vert =\sG^{-1}$.
This gives $\kappa=2n\sV^{2}\sG^{-2}$. Because
\[
a^{\dagger}b_{0}^{-1}a=-\sG^{-1}a^{\dagger}a
\]
we have $\mu_{0}=2n\sV^{2}\sG^{-1}$. Using $\mu\approx\mu_{0}$ we
find 
\begin{align*}
\alpha & =\min\left\{ \frac{1}{2}\left\Vert b^{-1}\right\Vert ^{-1},\,\frac{1}{4}\kappa^{-1}\mu\right\} \\
 & =\min\left\{ \frac{1}{2}\sG,\frac{1}{4}\frac{2n\sV^{2}\sG^{-1}}{2n\sV^{2}\sG^{-2}}\right\} \\
 & =\sG/4
\end{align*}
 and so 
\begin{align*}
\beta & =\max\left\{ 1,\,\alpha^{-1}\left\Vert b^{-1}\right\Vert ^{-1}\right\} \\
 & =4
\end{align*}
 and with Theorem~\ref{thm:relax-bound-1} we get the bounds
\begin{align*}
\Delta\tau & \leq\frac{20}{\pi}\sG^{-1}\\
\Delta\tau_{\mathrm{rel}} & \leq\frac{40}{\pi}n\sV^{2}\sG^{-2}\,.
\end{align*}

To get the bound on $\Delta\tau_{1,\,\mathrm{rel}}$ we also estimate
$\left\Vert \nu\right\Vert $, we use the fact that each column and
row of $\nu$ has $(n-2)$ nonzero entries and so
\[
\left\Vert \nu\right\Vert \geq vn\sV
\]
for a scaling and dimension independent constant $v$. Then Theorem~\ref{thm:relax-bound-2}
gives the bound
\[
\Delta\tau_{1,\,\mathrm{rel}}\leq4vn\sV\sG^{-1}\,.
\]
The condition for this bound is

\begin{align*}
\left\Vert \nu\right\Vert  & \leq\frac{1}{2}\left\Vert b^{-1}\right\Vert ^{-1}
\end{align*}
the LHS is bounded from below by $vn\sV$ and the RHS is constant,
so the condition does not hold for large $n$.

\subsection{Circular chain}

Assume the sites are positioned on a circle and only nearest neighbors
interact with strength $\sV$
\[
V_{kl}=\begin{cases}
\sV & |k-l|=1\\
0 & \mathrm{else}
\end{cases}
\]
where we set equivalence $n\equiv0$. Further $\gamma_{k}=\sG$ and
$E_{k}$ such that $E_{kl}=\sG E$ when $|k-l|=1$ which is possible
for $n$ even.

Now, the column for site $k$ in $a$ has only 4 entries, two each
for the coherences with $k-1$ and $k+1$. We calculate
\[
\left(a^{\dagger}a\right)_{kl}=\begin{cases}
4\sV^{2} & k=l\\
-2\sV^{2} & \left|k-l\right|=1\\
0 & \mathrm{else}
\end{cases}
\]
So $\left\Vert a^{\dagger}a\right\Vert =8\sV^{2}$ and $\left\Vert a\right\Vert =\sqrt{8}\sV$,
in particular there is no $n$ dependency. Also $\left\Vert b_{0}^{-1}\right\Vert =1/\sqrt{\sG^{2}+\sG^{2}E^{2}}$
and so $\kappa=\frac{8}{1+E^{2}}\sV^{2}\sG^{-2}$. We have
\[
N_{0}=\begin{cases}
-\frac{4\sV^{2}}{\sG(1+E^{2})} & k=l\\
\frac{2\sV^{2}}{\sG(1+E^{2})} & \left|k-l\right|=1\\
0 & \mathrm{else}
\end{cases}
\]
which has the spectrum
\begin{equation}
\lambda_{p}=-\frac{4\sV^{2}}{\sG(1+E^{2})}\left(1-\cos\left(\frac{2\pi p}{n}\right)\right)\label{eq:spectrum-linear-chain-1}
\end{equation}
with $p=1\dots n$. The nonzero eigenvalue smallest in magnitude is
$\mu_{0}$, so for large $n$ and small $\sV\sG^{-1}$, approximately
\[
\mu\approx\mu_{0}\approx\frac{2\sV^{2}}{\sG(1+E^{2})}\left(\frac{2\pi}{n}\right)^{2}
\]
 so $\alpha=\frac{1}{16}\sG\left(\frac{2\pi}{n}\right)^{2}$ and 
\begin{align*}
\beta & =\frac{16\sG\sqrt{1+E^{2}}}{\sG\left(\frac{2\pi}{n}\right)^{2}}\\
 & =\left(\frac{2n}{\pi}\right)^{2}\sqrt{1+E^{2}}\,.
\end{align*}
Moving the numbers into constants $k_{1}$ and $k_{2}$, and dropping
the $1$ in $1+\beta$ (fine for large $n$), we have
\begin{align*}
\Delta\tau & \leq k_{1}\sqrt{1+E^{2}}\sG^{-1}n^{4}\\
\Delta\tau_{\mathrm{rel}} & \leq\frac{k_{2}}{\sqrt{1+E^{2}}}\sV^{2}\sG^{-2}n^{2}\,.
\end{align*}
We again estimate $\left\Vert \nu\right\Vert $, now each column and
row of $\nu$ has $2$ or $4$ nonzero entries and so
\[
v_{1}\sV\leq\left\Vert \nu\right\Vert \leq v_{2}\sV
\]
for some scaling and dimension independent constants $v_{1}$ and
$v_{2}$. Then Theorem~\ref{thm:relax-bound-2} gives the bound
\[
\Delta\tau_{1,\,\mathrm{rel}}\leq\frac{4}{\pi^{2}}v_{2}n^{2}\sV\sG^{-1}\,.
\]
This time the condition

\begin{align*}
\left\Vert \nu\right\Vert  & \leq\frac{1}{2}\left\Vert b^{-1}\right\Vert ^{-1}
\end{align*}
does not break down for large dimensions, so the bound holds for all
$n$ when $\sV$ and $\sG$ are kept constant.

\bibliographystyle{plain}
\phantomsection

\end{document}